\documentclass[article]{gmp2014}

\usepackage{amsmath}
\usepackage{amsfonts}
\usepackage[mathscr]{euscript}
\usepackage{dsfont}
\usepackage{amssymb}
\usepackage{graphicx}
\usepackage{url}
\usepackage{rotating}
\usepackage{verbatim}
\usepackage{color}
\usepackage{svg}
\usepackage{phonetic}

\theoremstyle{definition}
\newtheorem{hypo}{\textsf{\textbf{{Hypothesis}}}}
\newtheorem{defn}{\textsf{\textbf{{Definition}}}}
\newtheorem{prop}{\textsf{\textbf{{Proposition}}}}

\DeclareFontFamily{OT1}{pzc}{}
\DeclareFontShape{OT1}{pzc}{m}{it}{<-> s * [1.10] pzcmi7t}{}
\DeclareMathAlphabet{\mathpzc}{OT1}{pzc}{m}{it}

\newcommand{\ii}{\text{\textsf{\textbf{i}}}}

\newcommand{\blue}[1]{\textcolor{blue}{#1}}
\newcommand{\nquad}{\!\! \!\! \!\!}
\newcommand{\nqquad}{\nquad\nquad}

\begin{document}

\title{Shape Complementarity Analysis for Objects of Arbitrary Shape\protect\footnotemark}

\author{Morad Behandish}{}
\author{Horea T. Ilie\c{s}}{}

\affiliation{}{\\ Departments of Mechanical Engineering and Computer Science and Engineering,
University of Connecticut, CT 06269}

\affiliation{}{\\ Technical Report No. CDL-TR-14-01, January 2014}{}


\maketitle

\footnotetext{This technical report was never published as a journal article. For citation, please use:
    \protect\\
    \protect\\
        \blue{Behandish, Morad and Ilie\c{s}, Horea T., 2014. ``Shape Complementarity Analysis for Objects of Arbitrary Shape.'' Technical Report No. CDL-TR-14-01, University of Connecticut.}
    \protect\\
    \protect\\
    A short presentation of this work will appear in the proceedings of the FWCG'2014 workshop \cite{Behandish2014b}.
    }

\begin{abstract}
    The basic problem of shape complementarity analysis appears fundamental to applications as diverse as mechanical design, assembly automation, robot motion planning, micro- and nano-fabrication, protein-ligand binding, and rational drug design. However, the current challenge lies in the lack of a general mathematical formulation that applies to objects of arbitrary shape. We propose that a measure of shape complementarity can be obtained from the extent of approximate overlap between shape skeletons. A space-continuous implicit generalization of the skeleton, called the skeletal density function (SDF) is defined over the Euclidean space that contains the individual assembly partners. The SDF shape descriptors capture the essential features that are relevant to proper contact alignment, and are considerably more robust than the conventional explicit skeletal representations. We express the shape complementarity score as a convolution of the individual SDFs. The problem then breaks down to a global optimization of the score over the configuration space of spatial relations, which can be efficiently implemented using fast Fourier transforms (FFTs) on nonequispaced samples. We demonstrate the effectiveness of the scoring approach for several examples from 2D peg-in-hole alignment to more complex 3D examples in mechanical assembly and protein docking. We show that the proposed method is reliable, inherently robust against small perturbations, and effective in steering gradient-based optimization.

\end{abstract}


\section{Motivation} \label{sec_motivation}

Research in computational shape analysis has impacted a wide range of disciplines from engineering to medicine, archeology, art, and entertainment. A quantitative description of shape complementarity, in particular, is critical to numerous applications, including:
\begin{itemize}
    \item CAD/CAM, design and assembly automation, from industrial scale to micro- and nano-structures;
    \item robotics, motion and path planning, obstacle avoidance, direct and inverse configuration space problems;
    \item systematic garment design systems, optimizing piecewise design and fabrication of apparel products;
    \item modeling biological processes involving molecular recognition, from gene expression to protein-ligand binding and drug design.
\end{itemize}
More generally, in almost every problem where a physical system is modeled by an arrangement of interacting pointsets, complementary geometrical features are indisputably critical determinants of proper assembly and function. However, in spite of its ubiquity and versatility, obtaining a general measure of shape complementarity for objects of arbitrary shape remains an open problem.

\begin{figure}
    \centering
    \includegraphics[width=0.48\textwidth]{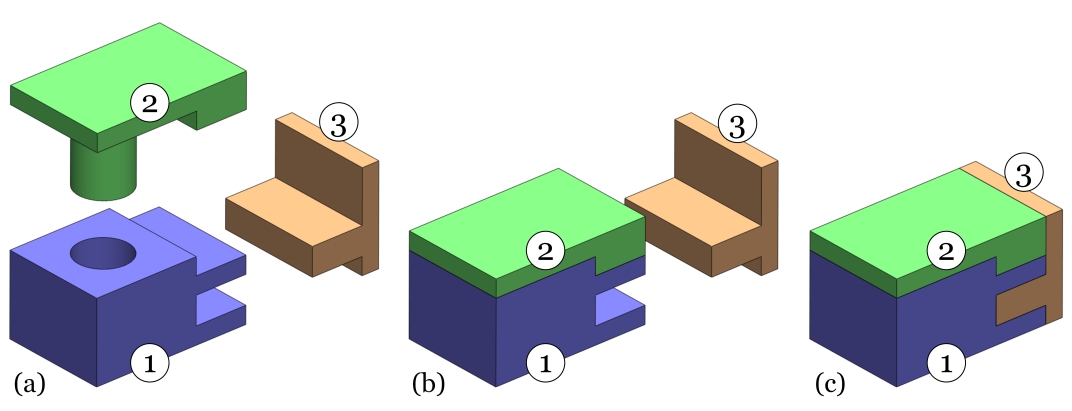}
    \caption{Automatic identification of spatial relations and assembly sequence is governed by shape complementarity - figure reproduced from \cite{Latombe1991} with minor modifications.} \label{figure1}
\end{figure}

A particularly important application that can benefit from such an analysis is assembly automation. Figure \ref{figure1} shows the assembly sequence of 3 simple mechanical components. An automatic identification of proper mating features without user interference turns out to be quite challenging, which reduces to an optimization of geometric complementarity. Once determined, the corresponding sequence of spatial relations between the properly assembled parts can be transferred to a description of robot motions for automatic assembly \cite{Latombe1991}.

Among other macro-scale applications are feature recognition for manufacturing \cite{Han2000}, motion planning in crowded or narrow environments \cite{Latombe1991}, fixture design \cite{Wu1998}, metal stamping die design \cite{Cheok1998}, and design of garment products \cite{Wang2005a}. However, shape complementarity is also central to the fabrication of micro-assemblies, including automated assembly of micro-mirrors \cite{Reid1998} and optical manipulation of micro-puzzles \cite{Rodrigo2007}. The impacts are realized even further in nano-fabrication where self-assembly is more critical. More recent studies on the placement of DNA origami structures \cite{Woo2011}, and other DNA modular assemblies composed of tetrahedra \cite{Goodman2008}, nano-tubes \cite{Aldaye2009}, and tensegrities \cite{Liedl2010}, further pronounce the demands for a fundamental study of the subject.

\begin{figure}
    \centering
    \includegraphics[width=0.48\textwidth]{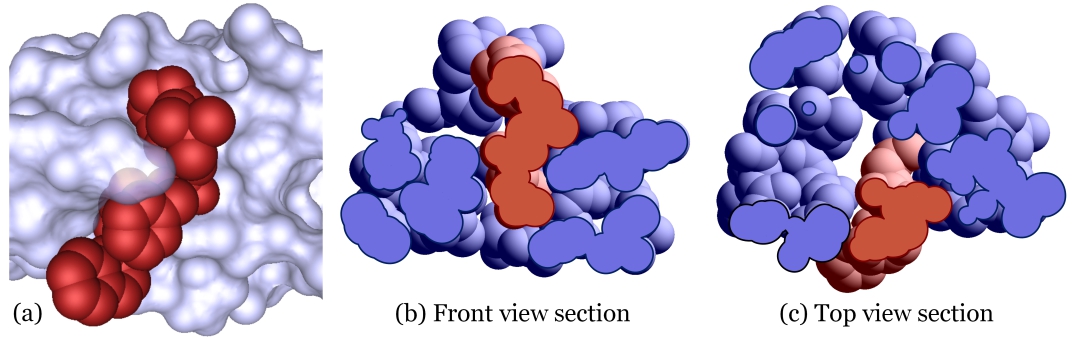}
    \caption{The binding site of HIV protease in complex with Saquinavir (PDB Code: 1FB7) \cite{Hong2000}; the shape complementarity is visually obvious, yet difficult to quantify.} \label{figure2}
\end{figure}

The significance of shape complementarity has been known to the molecular biologists since the earliest days of protein structure determination in 1950s \cite{Kendrew1958}. Proteins are the cellular machinery of living organisms, responsible for a vast array of biological functions, exclusively determined by their 3D structures. Geometric alignment has been identified as a vital player in different levels of structural studies; for instance, in folding of $\alpha-$helical domains (e.g., packing models such as `knobs in holes' and `ridges in grooves' \cite{Eilers2002}), and in complex formations (e.g., binding models such as `lock and key' or `induced fit' \cite{Koshland1995}). However, most of these models have been restricted to qualitative interpretations by the biologists.
One particularly significant application of binding prediction is in rational drug design, where geometric fit into the receptor protein is the main criterion in designing the scaffold of the `lead compounds' for drug molecules \cite{Kuriyan2012}.
Access to visualization software has facilitated this process for the pharmacologist, in a similar fashion that CAD technology has revolutionized the engineer's toolbox. However, automatic identification of a proper fit is even more helpful when dealing with molecular interfaces, noting that mating features are geometrically more complex and less clear to the human eye. Figure \ref{figure2} shows the drug molecule Saquinavir prescribed in HIV therapy, in complex with a subset of its receptor enzyme, HIV protease (PDB Code: 1FB7) \cite{Hong2000}.


\section{Related Work} \label{sec_relatedwork}

Most of the research in shape complementarity analysis is indebted to computational attempts in {\it ab initio} protein docking, extensively studied among computer scientists and computational biologists \cite{Ritchie2008a}.
The earlier techniques ranged from geometric hashing \cite{Lenhof1997,Duhovny2002}, to detection and matching of `knobs and holes' \cite{Wang1991,Shoichet1991,Lawrence1993}, and other signature features \cite{Jiang1991,Fischer1993,Duhovny2002} on the surfaces of the binding proteins. More recent methods characterize the topography of cavities and protrusions of the surface by means of a so-called `elevation function' \cite{Agarwal2006}, and compare the maxima of this function for geometric alignment of protein surfaces \cite{Wang2005}. A large sub-class of the most effective docking algorithms rely on correlations computed from overlapping geometric densities, reviewed in \cite{Eisenstein2004}. The earliest attempts rasterized the geometry to a uniform 3D grid, assigning different values to the cells based on their position with respect to the boundary. The grid representations of the mating partners were overlapped to obtain a score for a given relative ``pose'', and used simple methods for sampling and searching the space of all possible such poses \cite{Katchalski1992}. Some methods use discrete Fourier transforms (DFT) to speed-up the search with a cumulative computation of trial solutions, \cite{Gabb1997,Chen2003b,Kovacs2003,Kozakov2006}; while others use spherical harmonic Fourier transforms (SHFT) \cite{Max1988,Duncan1993,Ritchie2000,Ritchie2008}. These methods have evolved to today's sophisticated systems that use the same basic principles with better density models, grid-free representations, and fast search techniques that are based on adaptive sampling and nonequispaced DFTs \cite{Bajaj2011,Bajaj2013}. For a more complete review of different approaches, the reader is referred to \cite{Eisenstein2004,Ritchie2008a}.

In spite of the substantial amount of research on shape complementarity for arrangements of spherical atoms (i.e., proteins), the problem is scarcely studied for objects made of different primitives. Here we propose a novel formulation and computational framework for objects of arbitrary shape in 3D, potentially extensible to higher dimensions. 


\section{Our Approach} \label{sec_method}

Before proceeding to a rigorous formalism in Section \ref{sec_form}, let us consider the state-of-the-art in the class of protein docking algorithms that are based on overlapping geometric density functions. Given a receptor protein $M_1$ and a ligand molecule $M_2$ in the 3D space, the common idea behind the majority of these methods is to formulate the so-called shape complementarity score $f(\tau; M_1, M_2)$ as a function of the relative rigid body transformation $\tau \in \mathrm{SE}(3)$:
\begin{equation}
    f(\tau; M_1, M_2) = (\rho_1 \ast \rho_2) (\tau) = \int_{\mathds{R}^3} \rho_1(\mathbf{p})~\rho_2(\tau^{-1}\mathbf{p})~ dv, \label{eq_method_1}
\end{equation}
where $\ast$ stands for the cross-correlation operator, and $dv$ is the differential volume element at $\mathbf{p} \in \mathds{R}^3$. The problem then reduces to an efficient search of the configuration space $\mathrm{SE}(3)$ of rigid motions for a subset of near-optimal poses that maximize $f(\tau)$.
The main challenge, however, lies in a proper formulation of the so-called {\it affinity functions} $\rho_{1,2}(\mathbf{p}) = \rho(\mathbf{p}, M_{1,2})$ for the two molecules, whose cross-correlation gives a high score when the two are in a ``proper fit'' pose.
A relatively successful and contemporary formulation is based on the double-skin layer (DSL) approach, illustrated in Figure \ref{figure3} (a). It was first introduced in \cite{Connolly1986,Wang1991} in 2D and 3D, respectively, implemented using grid-based equispaced FFTs in \cite{Katchalski1992}, generalized to complex affinities in \cite{Chen2003b}, and evolved to a grid-free nonequispaced FFT-based algorithm in \cite{Candas2005,Bajaj2011}.
Two skin regions are defined implicitly: 1) the grown skin region of the receptor, formed by an additional layer of pseudo-atoms populated over the solvent accessible surface (SAS); and 2) the surface skin region of the ligand, formed by the original solvent accessible atoms.
The remaining atoms in both molecules are labeled as core atoms. A radial-basis affinity function was proposed in \cite{Candas2005,Bajaj2011} as a summation of Gaussian kernels centered at the skin/core atoms and pseudo-atoms, modeling the space occupancy of the molecules as a smooth 3D density function. Here is (a slight variation of) their formulation:
\begin{equation}
    \rho(\mathbf{p}, M) = \sum_{1 \leq j \leq n} c_j ~ g_\sigma \left( \frac{\| \mathbf{p} - \mathbf{q}_j \|_2}{R_j} - 1 \right), \label{eq_method_2}
\end{equation}
where $n = n_{1,2}$ is the number of atoms (including pseudo-atoms) in one molecule $M = M_{1,2}$; $\mathbf{p} \in \mathds{R}^3$ is the query point, $\mathbf{q}_j \in \mathds{R}^3$ and $R_j$ are the center and radius of the $j$'th atom whose contribution to the sum is weighted by $c_j$ for $j = 1, 2, \cdots, n$, $\|\cdot\|_2$ is the $L_2-$norm, and $g_\sigma(x) \propto \exp[-\frac{1}{2}(x/\sigma)^2]$ is the isotropic Gauss function, the decay rate of which is governed by the constant factor $\sigma > 0$.

The weights $c_j$ are usually chosen to be different constants for skin and core atoms. A clever trick was used by performing computations on complex functions, picking real coefficients $c_j \propto \lambda_1$ for skin atoms and imaginary coefficients $c_j \propto \ii \lambda_2$ ($\ii = \sqrt{-1}$) for core atoms (typically $\lambda_2 > \lambda_1$).
Subsequently, $f'(\tau) := \mathrm{Re}\{ f(\tau)\}$ or $f'(\tau) := \mathrm{Re}\{ f(\tau)\} - \mathrm{Im}\{ f(\tau)\}$ can be used as the actual score for optimization. As a result, skin-skin overlaps between the receptor and ligand impart a positive real {\it reward}, core-core collisions impart a negative real {\it penalty}, and skin-core collisions optionally contribute a smaller penalty to the total score. This was improved to a more sophisticated skin-curvature and core-depth dependent weight assignment in \cite{Bajaj2013}.
The configuration space $\mathrm{SE}(3) \cong \mathrm{SO}(3) \rtimes \mathrm{T}(3)$ is then searched for the highest values of $f'(\tau)$ using an adaptive sampling of rotations in $\mathrm{SO}(3)$; for each trial rotation, the translation space $\mathrm{T}(3) \cong \mathds{R}^3$ is searched rapidly using either a 3D equispaced FFT (the well-known Cooley-Tukey algorithm \cite{Cooley1965}), or a 3D nonequispaced FFT (Potts et al.'s algorithm \cite{Potts2001}).

The elegance of this approach lies in the exploitation of the radial-basis Gaussian representation in speeding up the search process; it preserves the shape and location of the score profile, compared to the expensive gird-based FFT approaches \cite{Eisenstein2004}, using a significantly less amount of memory and computation time (up to two orders of magnitude) \cite{Bajaj2006}.
However, these implementations have been progressively calibrated and specialized for the particular application in protein docking. There are a few caveats when trying to extend the ideas to a purely geometric assessment of the shape complementarity for arbitrary shapes; for instance, the radial-basis approach imposes the following limitations:
\begin{enumerate}
    \item It requires an additional step to identify the skin/core classification of original atoms; however, despite being well-defined for globular proteins, the distinction might not be clear for molecules of extended shape.
    \item It imposes the requirement to populate a finite number of pseudo-atoms along the solvent-accessible surface, which can be done in infinitely many different ways, resulting in different scoring schemes.
\end{enumerate}
The outcomes of these pre-processing steps impart inevitable bias to the scoring outcome; in fact the algorithm that populates the pseudo-atoms and its parameters (e.g., the packing density or grid spacing in the algorithm given in \cite{Candas2005}) are likely to affect the optimal results even more than those of the affinity formulation in (\ref{eq_method_2}), which is not desirable.
A possible alleviation to the bias is to overpopulate the surface and grown skins with many pseudo-atoms, allowing penetrations between pseudo-atoms, at the expense of increasing the input size. The bias eventually disappears if $n_{1,2} \rightarrow \infty$, without modifying the shape of the cavity occupied by the original molecules.
This can be conceptualized by sweeping a Gaussian probe along the inside and outside of the protein surface.

Let us rewrite (\ref{eq_method_2}) with the skin-core separation:

\begin{equation}
    \rho(\mathbf{p}, M) = \lambda_1 \!\! \sum_{\text{skin}} g_\sigma \left( \frac{\eta}{R} - 1 \right) + \ii \lambda_2 \!\! \sum_{\text{core}} g_\sigma \left( \frac{\eta}{R} - 1 \right), \label{eq_method_2_sep}
\end{equation}
where $\eta = \eta(\mathbf{p}, \mathbf{q}) = \| \mathbf{p} - \mathbf{q} \|_2$ denotes the distances from the skin or core atom centers to the query point.
For simplicity, the probe radius $R > 0$ is assumed to be the same for all skin and core atoms, respectively, but do not have to be.
The skin and core in this case are defined by a finite number of center points of the atoms in the original molecule or the added pseudo-atoms whose center points are sampled on the offset surfaces (e.g., SAS).
To eliminate the bias arising from the arbitrary sampling, the finite sum can be replaced with an integral that represents the sum in the limit as the number of pseudo-atoms approaches $\infty$:
\begin{equation}
    \rho(\mathbf{p}, M) = \lambda_1 \!\! \int_{\text{skin}} \nquad g_\sigma \left( \frac{\eta}{R} - 1 \right) ~ds + \ii \lambda_2 \!\! \int_{\text{core}} \nquad g_\sigma \left( \frac{\eta}{R} - 1 \right) ~ds, \label{eq_method_2_int}
\end{equation}
where $ds$ is the differential surface element. The skin and core in this case are not finite sets, but are continuum surfaces defined by positive and negative offsets of the protein surface, respectively, i.e., $\pm R-$iso-surfaces of the signed distance function.\footnote{Replacing the core by a negative offset Gaussian shell is for the sake of simplifying the formulation, though one could still use the original atom centers (whose locations are not arbitrary, unlike pseudo-atoms) as in the second term on the right-hand side of (\ref{eq_method_2_sep}).}

\begin{figure}
    \centering
    \includegraphics[width=0.48\textwidth]{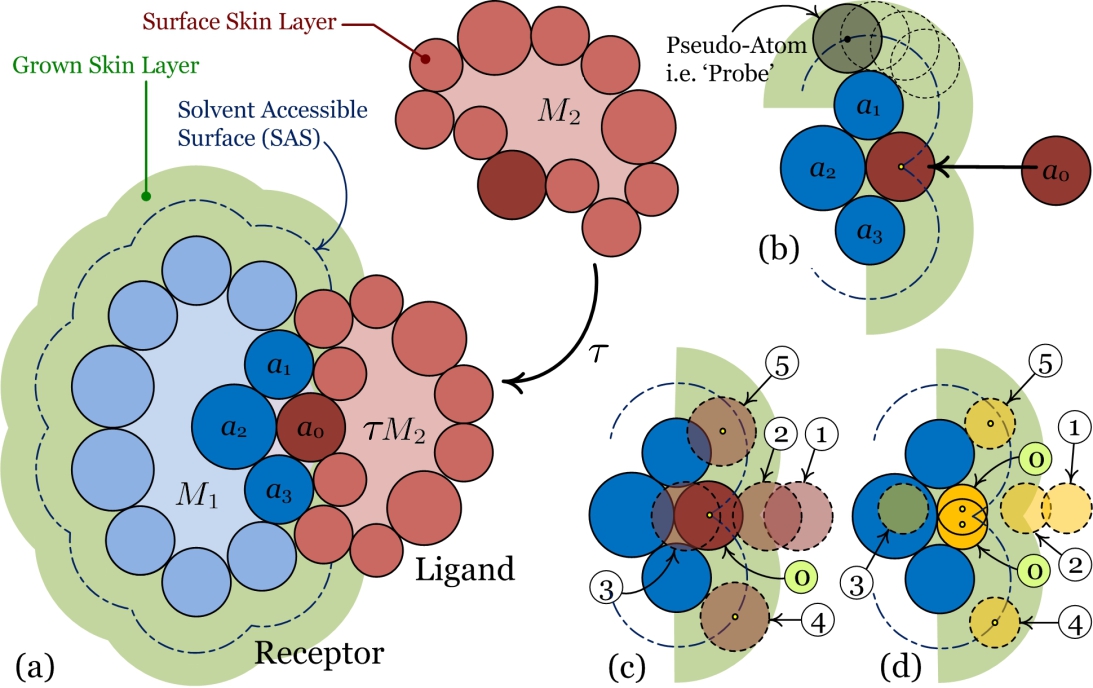}
    \caption{The double-skin layer technique for protein docking: the skin overlap does not discriminate between configurations of ligand protrusion with single and multiple contact points.} \label{figure3}
\end{figure}

A much simpler formulation can be obtained by directly applying the Gaussian kernel to the signed distance function, whose iso-surfaces were utilized; namely:
\begin{equation}
    \rho(\mathbf{p}, M) = \lambda_1 ~ g_\sigma \left( \frac{\xi}{R} - 1 \right) +\ii \lambda_2 ~ g_\sigma \left( \frac{\xi}{R} + 1 \right), \label{eq_method_3a}
\end{equation}
where $\xi = \xi(\mathbf{p}, M)$ is the signed distance function from each molecular surface and $R > 0$ is the probe radius as before. The first and second Gaussian terms create the swept-volume skin and core layers, respectively. This resembles the traditional skin-layer approach in \cite{Bajaj2011}, and can be adjusted to the improved skin-layer approach in \cite{Bajaj2013} where both the skin thickness $R$ and skin/core weights $\lambda_{1,2}$ are varied as a function of surface curvature and penetration depth, and the arguments of the Gaussian are offset to detach the skin from the original surface. Importantly, it is much easier to compute (\ref{eq_method_3a}) compared to (\ref{eq_method_2_int}) which requires explicit integration over offset surfaces.

The idea can be generalized to objects of arbitrary shape, at the expense of losing the speed and memory advantages of radial-basis formulation. There are several known inefficiencies, however, common to both radial-basis and swept-probe techniques, or their numerous possible variations, as long as they are based on overlapping skin/core layers in one way or another:
\begin{enumerate}
    \item The truncated scoring function has a compact support\footnote{The support of the scoring function $f(\tau): \mathrm{SE}(3) \rightarrow \mathds{C}$ is defined as $\mathrm{supp}(f) = \{ \tau \in \mathrm{SE}(3) ~|~ f(\tau) \neq 0 \}$. In this case, although the Gaussian kernel is not exactly zero at far distances, it is generally approximated with a truncated window function of compact support in the nonequispaced FFT implementation \cite{Potts2001}.} that spans an extremely narrow region of the configuration space, i.e., when the molecules are proximal. To the best of our knowledge, there is a lack of a proper formulation that is capable of providing strong local clues to guide the search algorithm, even when the objects are far away.
    \item Within the support region, there is little appreciation of the geometric features such as protrusion and depressions, sharp corners, and generally surface curvature similarities between the mating objects.
    \item The choice of a proper probe radius is not clear, esp. for objects of arbitrary shape: large probe size will miss small geometric features, while small probe size will further narrow down the support region.
\end{enumerate}
A na\"{\i}ve proposition to solve the first problem could be to replace the Gaussian with a different radial-basis kernel of slower decay rate---e.g., one of inverse-polynomial variation with distance, e.g., $\sim 1/\eta^c$ in (\ref{eq_method_2_sep}) and (\ref{eq_method_2_int}) or $\sim 1/\xi^c$ in (\ref{eq_method_3a}) with $c > 0$. The motivation comes from the electrostatic effect (Coulomb's law), which is the nature's means of steering molecules to bring them into proper contact \cite{Kuriyan2012}. Although it might help to guide the search towards the proper proximity, it also fades away the geometric specificity of the surfaces within the support region,\footnote{Note that nature also uses charge distributions, evolved to impart specificity, in addition to the geometric features.} and results in a flattened score function.

To better understand the other two problems and propose our solution, let us take a closer look at the group of atoms identified as $a_1$, $a_2$, and $a_3$ that form a depression on the receptor, which accommodates the protrusion made by atom $a_0$ on the ligand, shown in Figure \ref{figure3} (b). Alternative optimal and non-optimal positions of $a_0$ are illustrated in Figures \ref{figure3} (c) and (d), for different radii of $a_0$ compared to the probe radius. Horizontal movements from the optimal position (0) reduce the score due to either decreased skin-skin overlap (position (2)) or increased skin-core collision (position (3)), while the score is zeroed out to the right of position (1). However, the scoring is insensitive to sliding motions along the skin, giving the same score to positions (0) and (4) in panel (b), where the radius of $a_0$ happens to be equal to that of the probe. The scoring is less accurate if the skin was thicker, or $a_0$ was smaller, as in panel (c), in which case the algorithm favors position (4) over (5), although there is no contact at position 4, and gives the same lower score to positions (0) and (5).

At this point, it should be clear that a more objective and possibly more effective shape complementarity scoring function needs to take into account the {\it multiplicity} of contact, i.e., the number of contact points between the spheres in the simpler case of proteins, and more generally the topology and geometry of contact surface when assembling arbitrary objects. The superiority of position 0 in Figures \ref{figure3} (b) and (c) is realized by noting that $a_0$ is in contact with more than one atom from the partner molecule, i.e., its center is equidistant from multiple spherical sites, hence it lies on the interfaces of their Voronoi tessellation,\footnote{Here, the sites are spherical objects, hence the Voronoi diagram of the spheres, equivalent to Voronoi diagram of centers with additively weighted distance function (Johnson-Mehl model \cite{Johnson1939}) is of interest, which gives rise to hyperboloidical facets.}
and the multiplicity of contact is the same as the degree of the Voronoi $j-$faces $(0 \leq j \leq 2)$ (including edges and vertices), defined as the number of adjacent cells. This is illustrated for the same group of 4 atoms in Figure \ref{figure4} (a) through (c), for different radii of the ligand atom. An important observation is that the external subset of the Voronoi graph (i.e., union of $j-$faces) form the medial axis of the complementary space of the protein, while the atom centers form the most stable points of the medial axis of the interior; hence we propose the following:

\begin{hypo}
An effective model of geometric fit {\rm (i.e., shape complementarity)} for objects of arbitrary shape can be obtained from a comparative overlapping of shape skeletons between the mutually complement features.
\end{hypo}

\begin{figure}
    \centering
    \includegraphics[width=0.48\textwidth]{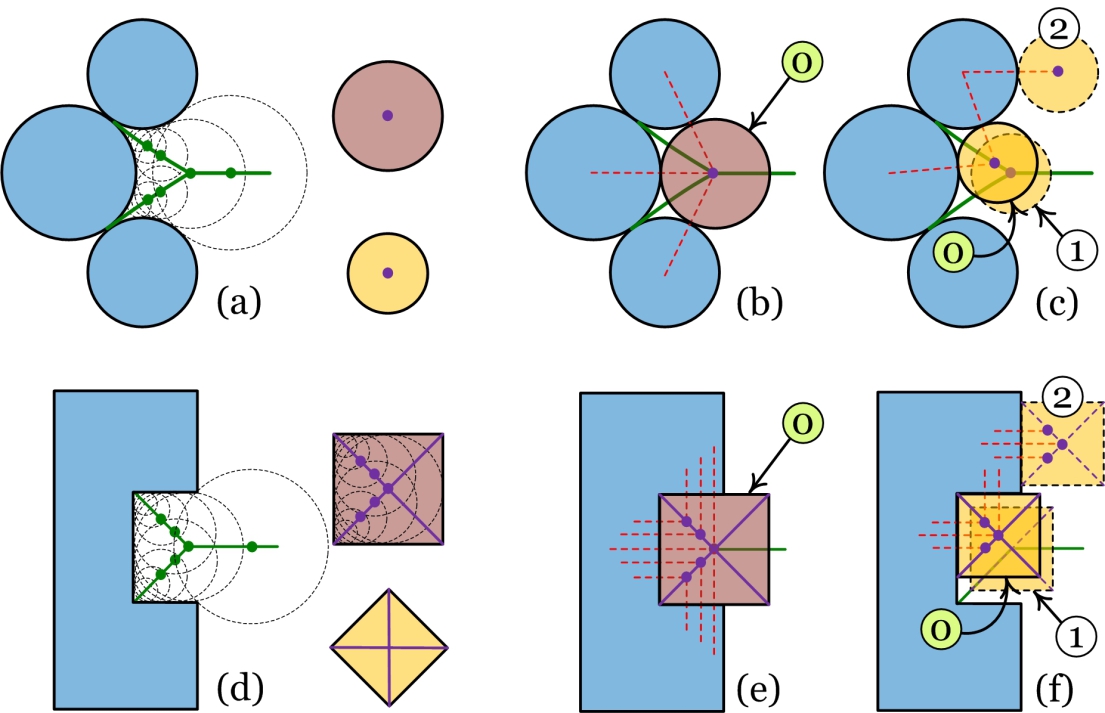}
    \caption{An overlapping of shape skeletons correlates with contact multiplicity (medial component), but proper contact also requires to reward proximity (proximal component).} \label{figure4}
\end{figure}

This can be achieved by overlapping the external skeleton of one object with the internal skeleton of its assembly partner. This is slightly harder to visualize when the primitives forming the objects are not spherical, as depicted for a simple case in Figure \ref{figure4} (d) through (f). The complementarity of geometric features, esp. those associated with sharp corners or local extremal curvatures, can be quantified by comparing the topology and geometry of surface contact between hypothetical balls that are swept to form the exterior of one object, with those that form the interior of its mating partner in a trial position. The subset of such balls that touch the boundary on multiple connected regions are the generator balls of the Medial Axis \cite{Choi1997}, suggesting the fundamental relationship between the skeletal topologies and shape complementarity. The usefulness of Medial Axes, as a compact abstraction of the topological and geometrical features of the object, is well known in the studies on shape similarities \cite{Tangelder2004, Torsello2004}; however, the following concerns are immediately observed:
\begin{enumerate}
    \item The conventional medial axis (and other variants of shape skeletons) are highly sensitive to noise/errors in shape data, making them extremely unstable and difficult to compute or approximate, even in 2D \cite{Lee1982, Attali2009}; and even harder to use for shape complementarity, esp. when assessing {\it approximate} complementarity.
    \item Obtaining a reliable measure of the extent of approximate overlap between the skeletons using their explicit representations is a challenge.
    \item Quantifying the multiplicity of contact, defined earlier as an integer when dealing with spherical primitives, is not a clear notion for general shapes.
\end{enumerate}
The difficulties associated with the cases of approximate complementarity are well-known (e.g., when the unbound objects are flexible, and the `induced fit' is made possible with deformations upon assembly) referred to as ``soft'' docking of proteins \cite{Jiang1991}. As illustrated in Figure \ref{figure5} for the simpler case of proteins, the overlaps between skeletons are not exact if the mating objects are not in perfect contact, demanding a flexible or ``weak'' definition for medial regions---see also Figure \ref{figure2} for a real example.

To solve this problem, we formulate an alternative representation of the shape skeleton, called the skeletal density function (SDF) that is implicit, continuous, and stable in the presence of errors, including noise in shape data or authentic shape discrepancies between the flexible parts before induced deformations during soft assembly.
Rather than sweeping a Gaussian probe for continuous double-skin layer generation as in (\ref{eq_method_2_int}), the Gaussian kernel could be assigned to the generator ball of the medial axis, e.g.,
\begin{equation}
    \rho(\mathbf{p}, M) = \lambda_1 \!\! \int_{\text{ext-MA}} \nqquad g_\sigma \left( \frac{\eta}{R} - 1 \right) ~ds + \ii \lambda_2 \!\! \int_{\text{int-MA}} \nqquad g_\sigma \left( \frac{\eta}{R} - 1 \right) ~ds, \label{eq_method_2_int_ma}
\end{equation}
where ext-/int-MA stand for external and internal medial axes, respectively, $\eta = \eta(\mathbf{p}, \mathbf{q}) = \|\mathbf{p} - \mathbf{q} \|_2$ is the distance to MA points, and $R \geq 0$ is the MA generator ball radius.
However, the medial axis is extremely difficult to compute and unstable under first and second order perturbations to surface geometry \cite{Attali2009} (which are common due to error/noise in data). In a similar fashion that  (\ref{eq_method_3a}) was obtained by applying the Gaussian kernel directly to the distance function and avoiding explicit computation of the offset surfaces in (\ref{eq_method_2_int}), it is desirable to avoid explicit computation of the medial axis.
This time, we apply the Gaussian kernel directly to the {\it distribution} of pairwise distances and its deviation from the minimum distance function to obtain an implicit generalization of the skeleton.

\begin{figure}
    \centering
    \includegraphics[width=0.48\textwidth]{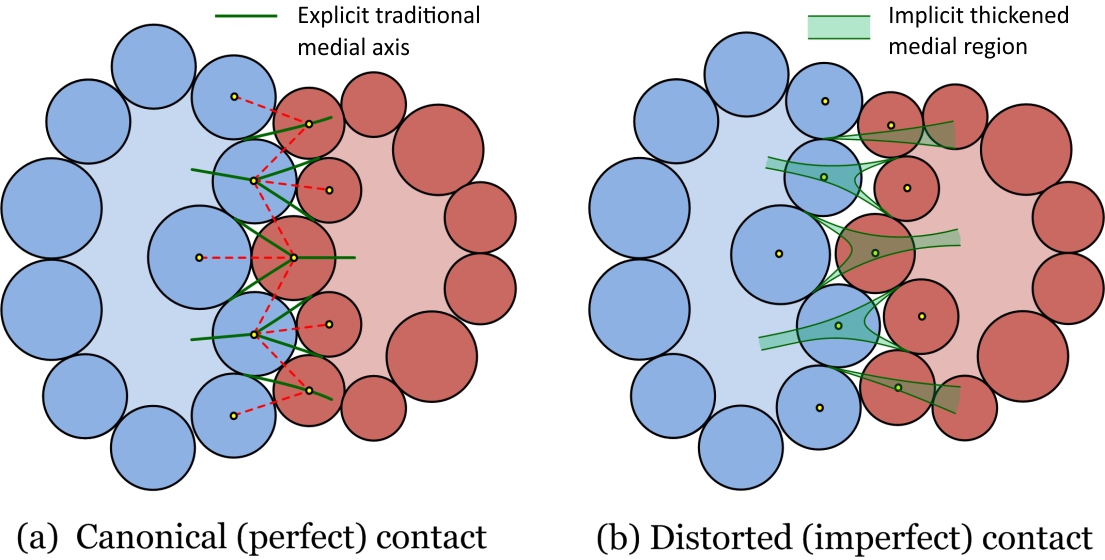}
    \caption{Traditional explicit skeletal diagrams are not sufficient to capture imperfect contact, hence a continuous implicit form is needed to give partial reward to approximate overlaps.} \label{figure5}
\end{figure}

In Section \ref{sec_form}, we propose an alternative formulation for the affinity functions $\rho_{1,2}(\mathbf{p})$ in the cross-correlation integral (\ref{eq_method_1}), which has as least two characteristic components:
\begin{enumerate}
    \item A `medial' component that characterizes the skeletal density, which will serve as a continuous measure of the multiplicity of contact with its generator balls.
    \item A `proximal' component that obligates the skeletal branches to stronger densities at the proximity of the surface, and decays as they move farther away.
\end{enumerate}
The first component results in a reward when the points of high skeletal density from the exterior of one object and the interior of its partner are overlapped in a trial pose - e.g., positions (0) and (1) in Figure \ref{figure4} are favored over position (2). The second component, on the other hand, favors proximity of the partner objects, hence ensures the superiority of position (0) over (1) in a proper trade off between the two effects. The combination of the two components results in an affinity function with a support region that resembles the branches of a ``thickened'' medial axis (see Figure \ref{figure5} (b)), slowly decaying away from the boundaries, yet possessing an intrinsic robustness to discrepancies, as a natural consequence of its definition as a density function that is continuous in both shape space and the $3-$space. We shall demonstrate that such a branched shape of the support region, together with the gradual decay of affinity, in turn results in a shape complementarity score function (\ref{eq_method_1}) that is capable of guiding the gradient-based search algorithms, moving the partner object along these skeletal branches toward the near-optimal fit position.


\section{Formulation} \label{sec_form}

The different problems ranging from mechanical assembly to molecular docking can be unified under a generic solid modeling representation.
Given two solids $S_1, S_2 \subset \mathds{R}^3$, the shape complementarity prediction problem can be divided into two separate steps:
\begin{enumerate}
    \item Formulating a shape complementarity {\it score function} $f(\tau; S_1, S_2)$ that quantifies their quality of fit, for a given relative transformation $\tau \in \mathrm{SE}(3)$.
    \item Developing an efficient algorithm to search the configuration space to obtain a subset of transformations $T \subset \mathrm{SE}(3)$ that result in affinities above a threshold $f(\tau) \geq f_0, ~ (\tau \in T)$; or proceeding further to find the optimal solution(s) $\tau^\ast \in T ~:~ f(\tau^\ast) = \max_{\tau \in T} \{ f(\tau) \}$.
\end{enumerate}
The restriction of the transformations to the special Euclidean group $\mathrm{SE}(3) \cong \mathrm{SO}(3) \rtimes \mathrm{T}(3)$ (i.e., combined proper rotations and translations) implies the assumption of rigidity for the mating shapes, which seems contradictory with our intention to implement soft assembly analysis. The flexibility is implicitly incorporated to some extent, however, by the choice of an affinity function that is resistant to small discrepancies between the shapes of the flexible parts---as in ``soft'' docking of proteins \cite{Jiang1991}.

It is desirable to obtain a score model that applies to objects of arbitrary shape, which is equivalent to a generic mapping $f: \mathrm{SE}(3) \times \mathcal{P}_\ast(\mathds{R}^3) \times \mathcal{P}_\ast(\mathds{R}^3) \rightarrow \mathds{R}$ (or $\mathds{C}$), which is a projection from the space of all possible shapes of the two solids and their relative transformations to the totally (or partially) ordered set of scores $\mathds{R}$ (or $\mathds{C}$). Here, $\mathcal{P}_\ast(\mathds{R}^3)$ stands for some proper subset of the power set of the $3-$space, often denoted $\mathcal{P}(\mathds{R}^3)$, to which we restrict our attention. On the other hand, all possible spatial relations between the two objects are identified as points $\tau \in \mathrm{SE}(3)$ that reside in the so-called configuration space \cite{Lozano-Perez1983}. The latter can be pictured by attaching coordinate frames to the objects, with respect to which the solids are described; then the configurations are the transformation that relate these frames by rigid motions \cite{Latombe1991,Joskowicz1999}.

For obvious practical reasons, we restrict our attention to solid objects, i.e., $\mathcal{P}_\ast(\mathds{R}^3)$ is chosen to be the collection of `r-sets', defined as bounded, closed, topologically regular, and semianalytic 3D sets whose boundaries (denoted $\partial S$) form piecewise smooth oriented $2-$manifold \cite{Shapiro2001}. Starting from the basic postulate that shape complementarity is concerned with the geometry of contact at the interface of the solids, the pointsets in the $3-$space that are relevant to this analysis can be reduced to those of the boundary manifolds. Such a formulation is particularly convenient for computations
when the two solids are specified using boundary representations (B-reps), which are widely popular in solid modeling \cite{Requicha1980}. The complex score function in (\ref{eq_method_1}) is repeated below for $f: \mathrm{SE}(3) \times \mathcal{P}_\ast(\mathds{R}^3) \times \mathcal{P}_\ast(\mathds{R}^3) \rightarrow \mathds{C}$:
\begin{equation}
    f(\tau, S_1, S_2) = (\rho_1 \ast \rho_2) (\tau) = \int_{\mathds{R}^3} \rho_1(\mathbf{p})~\rho_2(\tau^{-1}\mathbf{p})~ dv, \label{eq_form_8}
\end{equation}
where $\rho_{1,2}(\mathbf{p}) = \rho(\mathbf{p}, S_{1,2})$ stand for the shape-dependent affinity functions for $S_{1,2} \in \mathcal{P}_\ast(\mathds{R}^3)$, which is preserved under Euclidean isometries, i.e., $\rho(\mathbf{p}, \tau S) = \rho(\tau^{-1}\mathbf{p}, S)$. The rest of this Section is dedicated to formulate the universal affinity field that possesses the desired characteristics articulated in Section \ref{sec_method}.

We made the case in Section \ref{sec_method} that an effective affinity field extracts an abstraction of the shape in terms of the distribution of the minimum distance from the query point to the manifold, defined using the following notation:\footnote{In general, the distance between a point and a set is defined as the infimum (i.e., maximum lower-bound) of pairwise distances, since minimum may not exist (e.g., for open sets). However, for closed solids, the minimum exists and coincide with the infimum.}
\begin{equation}
    \xi(\mathbf{p}, S) = \mathrm{pmc}(\mathbf{p}, S) \min_{\mathbf{q} \in \partial S} \eta(\mathbf{p}, \mathbf{q}), \label{eq_form_10A}
\end{equation}
where $\eta(\mathbf{p}, \mathbf{q}) = \| \mathbf{p}- \mathbf{q} \|_2$ is the pairwise distance from boundary points $\mathbf{q} \in \partial S$ to the query point $\mathbf{p} \in \mathds{R}^3$ in which $\| \cdot \|_2$ denotes the $L_2-$norm, and $\mathrm{pmc}(\mathbf{p}, S) = +1$, $-1$, or $0$, is the point membership classification (PMC) of the query point as an external, internal, or boundary point, respectively \cite{Tilove1980a}. The PMC is equivalent to the knowledge of a consistent manifold orientation, e.g., one represented by a unit normal vector $\mathbf{n}(\mathbf{q}, S)$ for all $\mathbf{q} \in \partial S$ over the piecewise smooth manifold, and averaged on the sharp attachments.
The PMC is equivalent to a knowledge of a consistent manifold orientation, e.g., one represented by a unit normal vector $\mathbf{n}(\mathbf{q}, S)$ for all $\mathbf{q} \in S$ over a piecewise smooth manifold, and averaged on the sharp attachments.
Another definition that is useful in our formulation is the angle by which the orientation of the vector $(\mathbf{p} - \mathbf{q})$ connecting the boundary points $\mathbf{q} \in \partial S$ to the query point $\mathbf{p} \in \mathds{R}^3$ differs from the outward surface normal $\mathbf{n}(\mathbf{q}, S)$:
\begin{equation}
    \theta(\mathbf{p}, \mathbf{q}, S) = \cos^{-1} \left[ \frac{(\mathbf{p} - \mathbf{q})~}{\| \mathbf{p}- \mathbf{q} \|_2} \cdot \mathbf{n}(\mathbf{q}, S) \right], \label{eq_form_12a}
\end{equation}
where $\cdot$ stands for vector inner product.
The terminology in (\ref{eq_form_10A}) through (\ref{eq_form_12a}) is illustrated in 2D in Figure \ref{figure6} (a). Next, the concept of a maximal contact ball, pointed out in Section \ref{sec_method} with reference to the conventional medial axis, is extended as follows:

\begin{figure}
    \centering
    \includegraphics[width=0.48\textwidth]{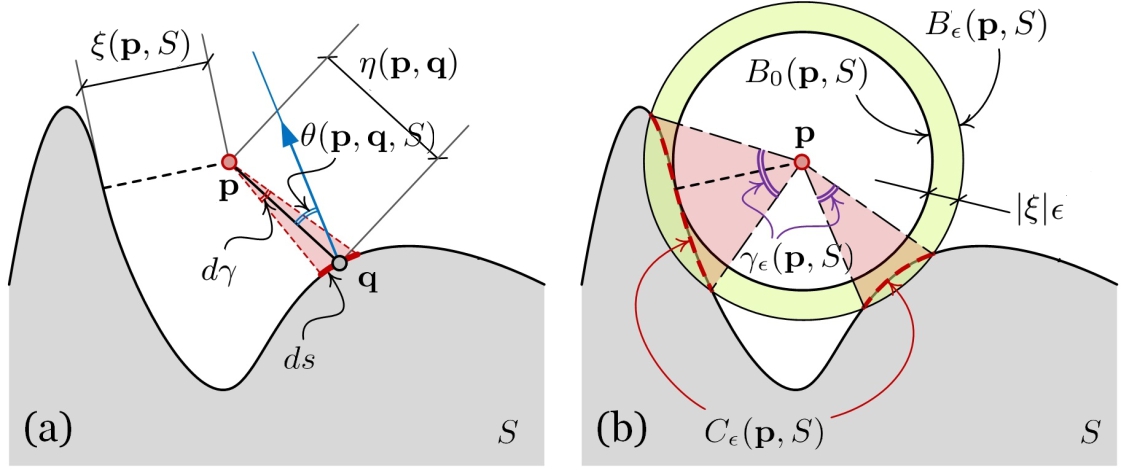}
    \caption{The SDF is formulated in terms of the distribution of distance from the query point to different boundary elements, integrated over the manifold region within the contact ball.} \label{figure6}
\end{figure}
\begin{defn}
    Given a solid $S \in \mathcal{P}_\ast(\mathds{R}^3)$, and a point $\mathbf{p} \in \mathds{R}^3$, the $\epsilon-$approximate contact ball (with the boundary) centered at $\mathbf{p}$ is defined as
    \begin{equation}
        B_\epsilon(\mathbf{p}, S) = \left\{ \mathbf{q} \in \mathds{R}^3 ~|~ \| \mathbf{p} - \mathbf{q} \|_2 \leq (1+\epsilon)|\xi(\mathbf{p}, S)| \right\}, \label{eq_form_10a}
    \end{equation}
    where $\epsilon \geq 0$ is an arbitrary constant. Subsequently, the $\epsilon-$approximate contact region is defined as the subset of the boundary $\partial S$ that lies within the contact ball:
    \begin{align}
        C_\epsilon(\mathbf{p}, S) &= B_\epsilon(\mathbf{p}, S) \cap \partial S \nonumber \\
        &= \left\{ \mathbf{q} \in \partial S ~|~ \| \mathbf{p} - \mathbf{q} \|_2 \leq (1+\epsilon)|\xi(\mathbf{p}, S)| \right\}. \label{eq_form_10b}
    \end{align}
    Note that contact regions are invariant (up to congruence) under Euclidean isometries, i.e., $C_\epsilon(\mathbf{p}, \tau S) \equiv C_\epsilon(\tau^{-1} \mathbf{p}, S)$ for all $\tau \in \mathrm{SE}(3)$, where $\equiv$ stands for congruence relation.\footnote{Two solids are congruent if one can be mapped to the other using a rigid transformation. In this case, $C_\epsilon(\mathbf{p}, \tau S) = \tau C_\epsilon(\tau^{-1} \mathbf{p}, S)$ thus $C_\epsilon(\mathbf{p}, \tau S) \equiv C_\epsilon(\tau^{-1} \mathbf{p}, S)$. Note that rigid transformations are homeomorphisms, thus congruent shapes have the same topological properties (e.g., number of connected components).}
\end{defn}
As pointed out in Section \ref{sec_method}, the {\it exact} contact ball $B_0(\mathbf{p}, S)$ corresponding to $\epsilon = 0$ is fundamentally related to the traditional definition of the shape skeleton. For the simple case of objects made of spherical primitives (i.e., proteins), $\mathbf{p} \in \mathds{R}^3$ is a $k$'th degree $j-$face ($1 \leq j \leq 2$) of the Voronoi diagram of the spheres if $|C_0(\mathbf{p}, S)| = k$, where $|\cdot|$ denotes set cardinality and $k \in \mathds{N} - \{1\}$ was referred to in Section \ref{sec_method} as the multiplicity of contact between $B_0(\mathbf{p}, S)$ and the boundary. For the case of arbitrary shapes, however, we argued that the notion of an integer multiplicity is not well-defined, since $|C_0(\mathbf{p}, S)|$ is not necessarily finite. In any case, by definition if $|C_0(\mathbf{p}, S)| > 1$, the query point lies on the medial axis.

Although the notions of medial axis and shape skeleton are not necessarily the same in the most general terms \cite{Choi1997}, both concepts are closely related to the geometry and topology of the contact region $C_0(\mathbf{p}, S)$, which could be composed of $0-$, $1-$ and $2-$dimensional components subject to a variety of possible degeneracies. The multiplicity of contact, in this case, can be defined as the number of connected components in $C_0(\mathbf{p}, S)$ (called `contact components' in \cite{Choi1997}), or that of $C_\epsilon(\mathbf{p}, S)$ as $\epsilon \to 0^+$, which need not be necessarily the same for all shapes.

Given two shapes $S_{1,2} \in \mathcal{P}_\ast(\mathds{R}^3)$ and a configuration $\tau \in \mathrm{SE}(3)$, a comparison of the number of contact components in $C_0(\mathbf{p}, S_1)$ and $C_0(\mathbf{p}, \tau S_2) \equiv C_0(\tau^{-1}\mathbf{p}, S_2)$ can identify if the two manifolds are in extensive exact contact with each other. However, a small perturbation of the manifold can invoke dramatic changes in the topology of $C_0(\mathbf{p}, S)$, making the medial axis and shape skeleton extremely unstable, and the shape complementarity score function (depending on it) inherently unreliable. On the other hand, when $\epsilon > 0$ is used in (\ref{eq_form_10a}) and (\ref{eq_form_10b}), the contact region is guaranteed to contain at least one $2-$dimensional region, as a result of the continuity condition on the manifold. The shape complementarity assessments are then built on the geometrical properties of $C_\epsilon(\mathbf{p}, S)$ (e.g., the extent of contact region), instead of the topological structure of $C_0(\mathbf{p}, S)$ (e.g., the number of contact components), to be implicitly embodied in the definition of the affinity function that follows.

Using the formulation in (\ref{eq_form_8}), all possible candidates for the affinity function $\rho: \mathds{R}^3 \times \mathcal{P}_\ast(\mathds{R}^3) \rightarrow \mathds{C}$ share at least two fundamental attributes, namely, 1) they all intend to extract the shape abstraction as real or complex fields over the $3-$space, hence are projections of the shape from a high-dimensional space $\mathds{R}^3 \times \mathcal{P}_\ast(\mathds{R}^3)$ into a 1D or 2D space (i.e., a single real or complex number per shape for a fixed point $\mathbf{p} \in \mathds{R}^3$); and 2) they all use the Euclidean geometry of the $2-$manifold in the $3-$space as observed from the query point $\mathbf{p} \in \mathds{R}^3$, which can be locally represented as a {\it distribution} of pairwise distances $\eta(\mathbf{p}, \mathbf{q}) = \| \mathbf{p} - \mathbf{q} \|_2$ from the boundary points to the query point. This leads to the following natural decomposition of the process to derive an affinity field $\rho: \mathds{R}^3 \times \mathcal{P}_\ast(\mathds{R}^3) \rightarrow \mathds{C}$:
\begin{enumerate}
    \item A projection $\zeta: \mathds{R}^3 \times \mathds{R}^3 \times \mathcal{P}_\ast(\mathds{R}^3) \rightarrow \mathds{C}$ that captures the essential distance distribution for the manifold as observed from $\mathbf{p}$.
    \item Applying a kernel $\phi: \mathds{C} \rightarrow \mathds{C}$ over the distance data to obtain the affinity $\rho(\mathbf{p}, S)$.
\end{enumerate}
The idea is schematically shown in Figure \ref{figure7}. The first step can be accomplished in a number of different ways. To obtain a visually meaningful and practically conducive layout, the following combination is proposed:

\begin{defn}
    Given a solid $S \in \mathcal{P}_\ast(\mathds{R}^3)$ and the points $\mathbf{p} \in \mathds{R}^3$ and $\mathbf{q} \in \partial S$, the complex projection $\zeta: \mathds{R}^3 \times \mathds{R}^3 \times \mathcal{P}_\ast(\mathds{R}^3) \rightarrow \mathds{C}$ is defined as:
    \begin{equation}
        \zeta(\mathbf{p}, \mathbf{q}, S) = \xi(\mathbf{p}, S) \pm \ii \eta(\mathbf{p}, \mathbf{q}), \label{eq_form_11}
    \end{equation}
    where $\mathrm{Re}\{\zeta\} = \xi(\mathbf{p}, S)$ and $\mathrm{Im}\{\zeta\} = \pm \eta(\mathbf{p}, \mathbf{q})$ were defined in {\rm (\ref{eq_form_10A})}, and the choice of $\pm$ is left arbitrary as a convention. Subsequently, the complex spread of the boundary as viewed from the query point is defined as:
    \begin{equation}
        \zeta(\mathbf{p}, \partial S, S) = \big\{ \zeta(\mathbf{p}, \mathbf{q}, S) ~|~ \mathbf{q} \in \partial S \big\}, \label{eq_form_12}
    \end{equation}
\end{defn}
\begin{figure}
    \centering
    \includegraphics[width=0.48\textwidth]{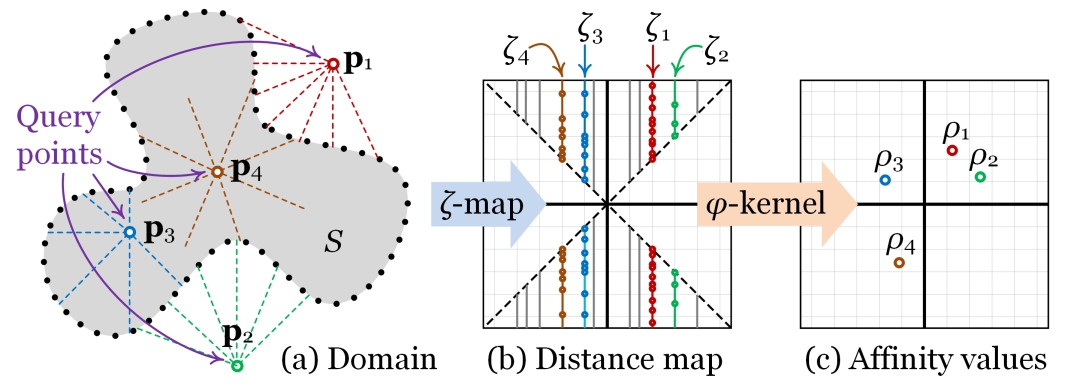}
    \caption{The affinity computation is decomposed into two steps: a projection in (\ref{eq_form_12}) characterizing distance distributions from the observer, followed by applying the $\phi-$kernel in (\ref{eq_form_13}).} \label{figure7}
\end{figure}
\begin{figure}
    \centering
    \includegraphics[width=0.48\textwidth]{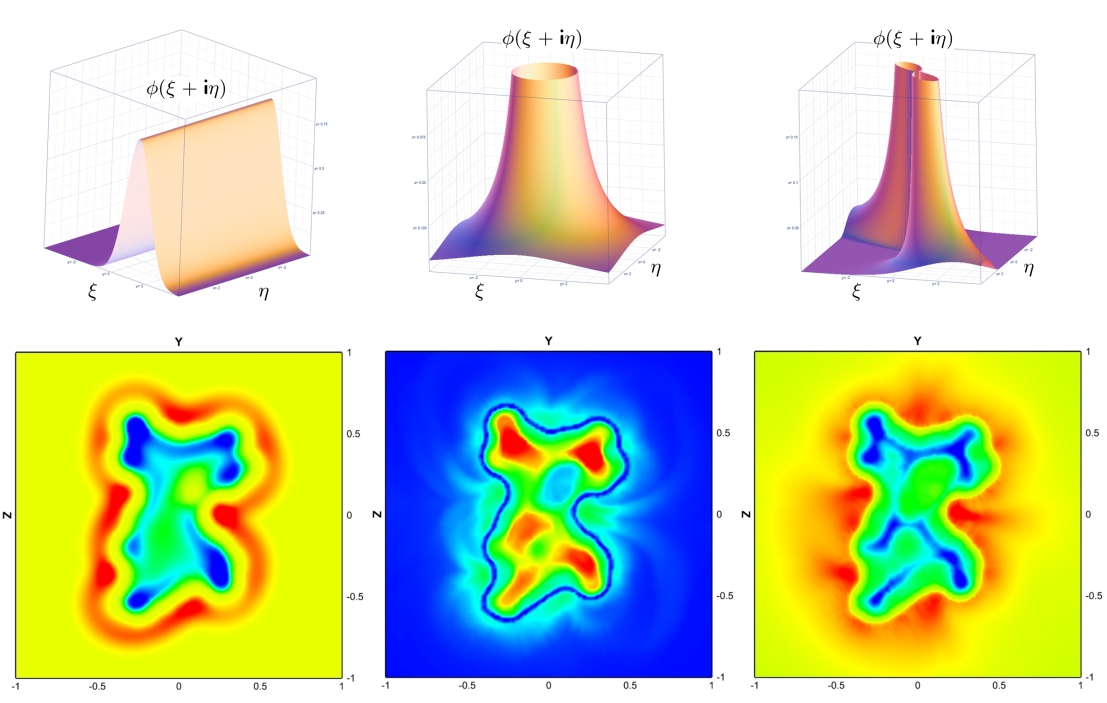}
    \caption{A few choices of the $\phi-$kernel (top) and a 2D section of the resulting affinity field for Saquinavir shown earlier in Figure \ref{figure2} (bottom). From left to right, $\phi(\zeta) \propto g_\sigma((\mathrm{Im}(\zeta)/R) - 1)$, $\phi(\zeta) \propto |\zeta|^{-2}$, and $\phi(\zeta) \propto |\zeta|^{-2} g_\sigma ( |\tan \angle \zeta| - 1)$. The last one is the kernel of choice in (\ref{eq_form_15}).} \label{figure8}
\end{figure}

The $\zeta-$map can be regarded as a projection of the shape into the complex plane from the viewpoint of a fixed observer $\mathbf{p} \in \mathds{R}^3$, retaining the distances $\|\mathbf{p} - \mathbf{q}\|_2$ while the orientation of $(\mathbf{p} - \mathbf{q})$ is lost. However, different projections at different query points collectively form a complete representation of the boundary; for instance, the manifold is implicitly specified as the signed distance iso-surface $\partial S = \{ \mathbf{p} \in \mathds{R}^3 ~|~ \mathrm{Re}\{\zeta(\mathbf{p}, \mathbf{q}, S)\} = 0 \}$, given any fixed boundary point $\mathbf{q} \in \partial S$.
Therefore, the $\zeta-$map contains all the information necessary to reproduce the shape (i.e., is informationally complete), but provides a representation that facilitates defining the affinity function.

Figure \ref{figure7} schematically shows the $\zeta-$map applied to a simple planar solid, where the complex spread $\zeta(\mathbf{p}, \partial S, S) \subset \mathds{C}$ is sketched for 3 query points $\mathbf{p}_1$, $\mathbf{p}_2$, and $\mathbf{p}_3$, each against a number of boundary points. It is clear from the definition in (\ref{eq_form_11}) that iso-surfaces of the signed distance function $\xi(\mathbf{p}, S)$ in the domain shown in panel (a), are each mapped to a closed vertical line segment in the complex plane, shown in panel (b); in fact a single query point $\mathbf{p}$ alone is mapped into a set of points $\zeta(\mathbf{p}, \partial S, S) = \xi(\mathbf{p}, S) + \ii \eta(\mathbf{p}, \partial S)$ spread along the aforementioned line segment, where each complex point $(\xi, \eta)$ corresponds to one or more ordered pair $(\mathbf{p}, \mathbf{q}) \in (\mathds{R}^3 \times \partial S$) that are a distance $\eta$ apart. The interior and exterior points are mapped to the left and right half-planes, respectively. The range is restricted to the phase angles $\pi/4 \leq \angle \zeta \leq 3\pi/4$ (where $\tan\angle\zeta = (\eta/\xi)$), since $\eta \geq |\xi|$ as a result of the definition in (\ref{eq_form_10A}), and the equality $\eta = |\xi|$ is exclusive to the pairs for which $\mathbf{q}$ is a nearest neighbor of $\mathbf{p}$ on the object surface, i.e., $\mathbf{q} \in C_0(\mathbf{p}, S)$. Furthermore, if $\mathbf{q} \in C_\epsilon(\mathbf{p}, S)$ for $\epsilon > 0$, as schematically depicted in panel (c), then $\zeta(\mathbf{p}, \mathbf{q}, S)$ lies in the closed region between the lines $\eta = |\xi|$ and $\eta = (1+\epsilon)|\xi|$, shaded in panel (b). For a given observer $\mathbf{p}$, the extent of the surface encapsulated within $C_\epsilon(\mathbf{p}, S)$ normalized with $\eta^2$---i.e., measured in terms of the {\it spatial angle}, as shown in panel (d)---is characterized in terms of the density of the layout within the phase angles $1\leq |\tan \angle \zeta| \leq (1+ \epsilon)$. This, in turn, can be magnified by a proper phase-dependent kernel. Subsequently, the affinity function is defined as follows:

\begin{figure*}
    \centering
    \includegraphics[width=1.00\textwidth]{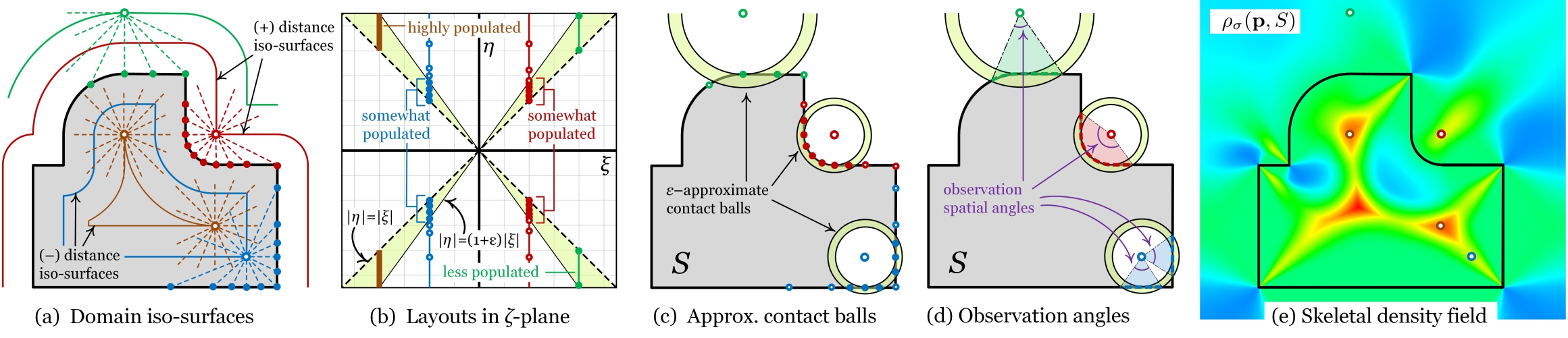}
    \caption{When a given point in the $3-$space is equidistant to an extensive region on the boundary to an $\epsilon-$approximation, it results in a $\zeta-$projection of the shape that is highly populated around the bisector $|\tan \angle \zeta| = 1$, hence a high skeletal density.} \label{figure9}
\end{figure*}

\begin{defn}
    Given a solid $S \in \mathcal{S}$, and a point $\mathbf{p} \in \mathds{R}^3$, the affinity function $\rho: \mathds{R}^3 \times \mathcal{S} \rightarrow \mathds{C}$ is defined as follows:
    \begin{equation}
        \rho(\mathbf{p}, S) = \oint_{\partial S} (\phi \circ \zeta)(\mathbf{p}, \mathbf{q}, S) ~\cos \theta(\mathbf{p}, \mathbf{q}, S) ~ d s, \label{eq_form_13}
    \end{equation}
    where $\zeta(\mathbf{p}, \mathbf{q}, S)$ was defined in {\rm (\ref{eq_form_11})}, $\theta(\mathbf{p}, \mathbf{q}, S)$ was defined in {\rm (\ref{eq_form_12a})}, and $d s$ is the differential surface element at $\mathbf{q} \in \partial S$.
\end{defn}
The choice of the kernel $\phi: \mathds{C} \rightarrow \mathds{C}$ determines the share of the points on the complex map $\zeta(\mathbf{p}, \partial S, S)$ to the integral in (\ref{eq_form_13}), thus should be chosen properly to impart the desired `medial' and `proximal' components, discussed in Section \ref{sec_method} into the affinity function. The term $\cos \theta$ is provided to assign each boundary element a weight of $d s_\bot = \cos \theta d s$, that is, the infinitesimal surface area perpendicular to the line of observation $(\mathbf{p} - \mathbf{q})$.

Before proceeding to propose a conceptually appealing and experimentally effective from for the $\phi-$kernel, it is worthwhile noting that (\ref{eq_form_13}) in its general form represents a platform for a variety of different scoring schemes, unified under a single general formulation. The differences between these methods would then be abstracted only in their particular choice of the $\phi-$kernel, as a result of the process decomposition illustrated in Figure \ref{figure7}. For instance, for a characteristic radius of $R > 0$, a choice of $\phi(\xi + \ii \eta) \propto g_\sigma(\frac{\eta}{R} - 1)$ or $\phi(\xi + \ii \eta) \propto g_\sigma(\frac{\xi}{R} \pm 1)$---using only one of $\xi = \mathrm{Re}\{ \zeta\}$ or $\eta = \mathrm{Im}\{ \zeta\}$ at a time---result in affinity functions that resemble the finite radial-basis or swept-volume double-skin layer affinity functions discussed in Section \ref{sec_method} (Figure \ref{figure8} (a)). To expand the narrow support region, one could try a more slowly decaying kernel, e.g., of inverse-square form $\phi(\zeta) \propto |\zeta|^{-2}$, which eliminates the characteristic size $R$, at the expense of less geometric specificity (Figure \ref{figure8} (b)). Figure \ref{figure8} illustrates a few examples of $\phi-$kernels and the resulting $\rho-$function.
We use a combination of ideas presented in Section \ref{sec_method} to devise a kernel that makes use of both $|\zeta| = (\xi^2 + \eta^2)^\frac{1}{2}$ and $\angle \zeta = \eta /\xi$ to provide a more effective shape descriptor that captures the essential properties of the distance distributions:

\begin{defn}
    We propose the kernel $\phi(\zeta): (\mathds{C} - \mathds{I}) \rightarrow \mathds{C}$:
    \begin{equation}
        \phi_\sigma(\zeta) = \frac{\lambda(\zeta)}{\sqrt{2\pi}} ~\frac{1}{\zeta^2}~ g_\sigma \left( |\tan \angle \zeta| - 1 \right), \label{eq_form_15}
    \end{equation}
    where $g_\sigma(x) = (\sqrt{2\pi}\sigma)^{-1} \exp[-\frac{1}{2}(x/\sigma)^2]$ is the isotropic Gauss function, and $\sigma > 0$ is the thickness factor. The coefficient $\lambda: (\mathds{C} - \mathds{I}) \rightarrow \{+\lambda_1, -\lambda_2\}$ is chosen as:
    \begin{equation}
        \lambda(\zeta) =
        \begin{cases}
            +\lambda_1 & \text{if } \mathrm{Re}\{ \zeta \} > 0,\\
            -\lambda_2 & \text{if } \mathrm{Re}\{ \zeta \} < 0.
        \end{cases}\label{eq_form_16c}
    \end{equation}
\end{defn}
The kernel in (\ref{eq_form_15}) is discontinuous on the imaginary line $\mathds{I} = \ii\mathds{R}$, and is unbounded around $\zeta = 0$ (i.e., when the query point approaches the boundary); however, the integral in (\ref{eq_form_13}) is still convergent at the vicinity of the boundary, since $\tan \angle \zeta \rightarrow \pm\infty$ hence $\phi(\zeta) \rightarrow 0$, when $\xi/\eta \rightarrow 0$. Figure \ref{figure8} (c) plots the kernel, and the resulting density function on the drug molecule in Figure \ref{figure2}.

As discussed in the introduction, the kernel is composed of a `medial' component $\phi_M(\zeta) = g_\sigma(|\tan \angle \zeta| - 1)$, a `proximal' component $\phi_P(\zeta) = 1/(\sqrt{2\pi} \zeta^2)$, and a coefficient $\lambda(\zeta) = \pm \lambda_{1,2}$.
The medial component $\phi_M$ takes the ratio $|\eta/\xi| = |\tan \angle \zeta|$ as the argument, giving rise to linear iso-curves $\phi_M = \text{\rm const}$, with the highest value of $\phi_M = (\sqrt{2\pi} \sigma)^{-1}$ when $\tan \angle \zeta = \pm 1$ on the bisectors $|\xi| = |\eta|$, and highly localized around it if $\sigma \ll 1$. This manifests in the affinity function in (\ref{eq_form_13}) as a higher density assignment to the query point $\mathbf{p} \in \mathds{R}^3$ that has an extensive set of approximately closest boundary elements (compare positions (0) and (1) in Figure \ref{figure4}).
The proximal component $\phi_P$, on the other hand, is provided to ensure contact between the assembly partners when their affinities are convolved in (\ref{eq_form_8}), as clarified in Section \ref{sec_method}, by introducing an inverse-square decay to the integrand in (\ref{eq_form_13}) (compare positions (0) and (1) in Figure \ref{figure4}.) In contrast to $\phi_M$ that was phase-dependent, the argument of $\phi_P$ is the amplitude $|\zeta|$, giving rise to circular iso-curves $\phi_P = \text{const}$, with $\lim_{\zeta \rightarrow 0} |\phi_P| = +\infty$.
The proposed $\phi-$kernel is singular at $\zeta = 0$, partly due to $|\phi_P| \rightarrow \infty$, and also because $\phi_M$ is undefined due to ambiguous $\angle \zeta$, and $\lambda$ is discontinuous. Therefore, the affinity is undefined for $\mathbf{p} \in \partial S$; however, it can be verified that $\rho(\mathbf{p}, S) \rightarrow 0^+$ when the query point $\mathbf{p}$ approaches the boundary ($\xi(\mathbf{p}, S) \rightarrow 0^\pm$), and the cross-correlation $f(\tau; S_1, S_2) = (\rho_1 \ast \rho_2) (\tau)$ is convergent for piecewise smooth boundaries.

Substituting for $\phi(\zeta)$ from (\ref{eq_form_15}) in (\ref{eq_form_13}) results in an implicit definition of a space-continuous SDF in the $3-$space, the high-density regions of which for $\sigma \ll 1$ resemble the $j-$faces $(0 \leq j \leq 2)$ of conventional skeletons:
\begin{equation}
    \rho_\sigma(\mathbf{p}, S) = \frac{\lambda}{\sigma} \oint_{\partial S} \frac{e^{-\frac{1}{2}(|\tan \varphi| - 1)^2/\sigma^2 - 2 \ii \varphi}}{(\cot^2{\varphi} + 1)/2} ~ \frac{d s_\bot}{4\pi \eta^2}, \label{eq_form_17}
\end{equation}
where $d s_\bot := \cos \theta ~ d s$ is the projected infinitesimal signed surface element, in which $\theta$ is the short form for the gaze angle $\theta(\mathbf{p}, \mathbf{q}, S)$ defined in (\ref{eq_form_12a}), and $\varphi$ is the short form for the phase of $\zeta-$map $\varphi(\mathbf{p}, \mathbf{q}, S) = \angle \zeta = \tan^{-1} (\eta/\xi)$. The interpretation of (\ref{eq_form_17}) is illustrated in Figure \ref{figure9}.

The largest contributions to the integral are due to the points with $|\tan \varphi| \approx 1$, at which the amplitude of the first fraction of the integrand evaluates to almost unity; the remaining term $d \gamma := \eta^{-2} ~d s_\bot = \eta^{-2} \cos \theta ~d s$ is the differential signed spatial angle by which the query point $\mathbf{p} \in \mathds{R}^3$ observes the surface element at $\mathbf{q} \in \partial S$, normalized by the term $4\pi$, clarifying part of the reason behind the inverse-square form used in (\ref{eq_form_15}). This can also be viewed as a surface integration over a sphere of radius $\eta$ centered at the query point $\mathbf{p}$ (namely, $\partial B_\epsilon(\mathbf{p}, S)$ for $\epsilon := |\eta/\xi| - 1$) normalized by the sphere surface area $4\pi \eta^2$.

On the other hand, when $|\tan \varphi|$ deviates significantly from unity, the Gaussian nominator vanishes. As a result, the density function can be approximated to an integration over the $\epsilon-$approximate contact region $C_\epsilon(\mathbf{p}, S)$ defined in (\ref{eq_form_10b}), where $\epsilon / \sigma$ determines the error. This leads to the following definition for the `truncated' affinity function:
\begin{equation}
    \bar{\rho}_{\sigma, \epsilon}(\mathbf{p}, S) = \frac{\lambda}{\sigma} \int_{C_\epsilon(\mathbf{p}, S)} \nquad \!\!\!\! \frac{e^{-\frac{1}{2}(|\tan \varphi| - 1)^2/\sigma^2 - 2 \ii \varphi}}{(\cot^2{\varphi} + 1)/2} ~ \frac{d s_\bot}{4\pi \eta^2}, \label{eq_form_17t}
\end{equation}
where $C_\epsilon(\mathbf{p}, S)$ is defined in (\ref{eq_form_10b}). The function $\bar{\rho}_{\sigma, \epsilon}(\mathbf{p}, S)$ approximates $\rho_\sigma(\mathbf{p}, S) = \lim_{\epsilon \to \infty} \bar{\rho}_{\sigma, \epsilon}(\mathbf{p}, S)$ , with an error of $|E_{\sigma, \epsilon}(\mathbf{p}, S)| = |\rho_\sigma(\mathbf{p}, S) - \bar{\rho}_{\sigma, \epsilon}(\mathbf{p}, S)| \ll 1$ if $\epsilon/\sigma \gg 1$.
Note also that the truncated affinity function $\bar{\rho}_{\sigma, \epsilon}(\mathbf{p}, S)$ could be obtained directly from (\ref{eq_form_13}) by using a truncated kernel $\bar{\phi}_{\sigma, \epsilon}(\zeta)$, which, in turn, is defined using a similar formula to that of (\ref{eq_form_15}), in which the Guassian is substituted with the truncated Gaussian defined as $\bar{g}_{\sigma, \epsilon}(x) = g_\sigma(x)$ if $|x| \leq \epsilon$, and $\bar{g}_{\sigma, \epsilon}(x) = 0$ otherwise. This introduces, for instance, an error of $< 10^{-4} \times g_\sigma(0)$ if $\epsilon > 4.3 \sigma$.

Using (\ref{eq_form_17t}) instead of (\ref{eq_form_17}) restricts the computations to the $\epsilon-$approximate nearest neighbors ($\epsilon-$ANN) of the query point, which is a much smaller set than the entire boundary. This significantly speeds up the implementation explained in Section \ref{sec_implement} with a negligible compromise in the accuracy due to neglecting the tail of the Gauss function.\footnote{Moreover, one can argue that the precise choice of a $\phi-$kernel was arbitrary to some extent, and there is no reason not to pick $\bar{\phi}_{\sigma, \epsilon}$ over $\phi_{\sigma}$ other than the elegance of simpler formulation for the latter, while better performance is achievable with the former.}

To make the tradeoff precise, (\ref{eq_form_17}) is first rearranged into a 1D integral, noting that the variable of integration can be changed to $r := |\tan \varphi| = |\eta/\xi| = (1 + \epsilon)$:

\begin{prop}
    If the $\phi-$kernel obeys the inverse-square decay law $\phi(\alpha \zeta) = \alpha^{-2} \phi(\zeta)$ for all $\alpha \in \mathds{R}$, the affinity integral can be rearranged into the following 1D form:
    \begin{equation}
        \rho(\mathbf{p}, S) = \int_1^{|\frac{\mu}{\xi}|} \gamma_{r-1}'(\mathbf{p}, S) ~ \phi \big( \mathrm{pmc}(\mathbf{p}, S) \pm \ii r \big) ~r^2 ~ dr, \label{eq_form_17A}
    \end{equation}
    where $\xi = \xi(\mathbf{p}, S)$ and $\mu = \mu(\mathbf{p}, S)$ are the min/max distances to the boundary, respectively, and $\gamma_\epsilon'(\mathbf{p}, S) := \frac{\partial~}{\partial \epsilon}\gamma_\epsilon(\mathbf{p}, S)$ in which $\gamma_\epsilon(\mathbf{p}, S)$ is the total spatial angle by which the contact region is observed from $\mathbf{p} \in \mathds{R}^3$:
    \begin{equation}
        \gamma_\epsilon(\mathbf{p}, S) = \int_{C_\epsilon(\mathbf{p}, S)} ~\frac{d s_\bot}{4\pi \eta^2} = \int_{C_\epsilon(\mathbf{p}, S)} \cos \theta ~\frac{d s}{4\pi \eta^2}.
    \end{equation}
    As before, $\mathrm{pmc}(\mathbf{p}, S) = +1$, $-1$, or $0$, while the choice of $\pm$ in $\pm \ii r$ comes from that of (\ref{eq_form_11}).
\end{prop}
\begin{proof}
    Let $r := |\eta/\xi|$ to obtain $\zeta = (\xi + \ii \eta) = \xi(1 \pm \ii r) = |\xi|(\mathrm{pmc}(\mathbf{p}, S) \pm \ii r)$, then $\phi(\zeta) = \phi(\mathrm{pmc}(\mathbf{p}, S) \pm \ii r) / \xi^2$ as a result of the inverse-square decay property. Also $\cos \theta~ds = \eta^2~d\gamma$ where $d\gamma$ is the differential signed spatial angle by which the surface element is observed from the query point; resulting in the term $r^2 = (\eta/\xi)^2$.
    On the other hand, $r = |\eta/\xi|$ is constant over the differential surface encapsulated within the spherical shell $C_{\epsilon + \delta} (\mathbf{p}, S) - C_{\epsilon}(\mathbf{p}, S)$ of infinitesimal thickness $|\xi|\delta$ as $\delta \to 0^+$, letting the shell inner and outer radii be $r = (1 + \epsilon)$ and $r + dr = (1 + \epsilon + \delta)$ $\delta \to 0^+$. This differential surface area can be rewritten as the total differential $d\gamma = \gamma_{(r - 1) + dr} - \gamma_{(r - 1)}$. Noting that
    \begin{equation}
        \frac{\partial ~}{ \partial \epsilon} \gamma_\epsilon(\mathbf{p}, S) \big|_{\epsilon = (r - 1)} = \lim_{\delta \to 0^+} \frac{\gamma_{\epsilon} - \gamma_{\epsilon + \delta}}{\delta} \big|_{\epsilon = (r - 1)},
    \end{equation}
    the surface area is obtained as $\frac{\partial \gamma_\epsilon}{ \partial \epsilon} ~dr$. Thus the 2D integral over the surface is converted to a 1D integral over the interval $\epsilon = (r-1) \in [0, |\mu/\xi|-1]$ for which $C_\epsilon(\mathbf{p}, S) \neq \emptyset$ thus $\frac{\partial \gamma_\epsilon}{\partial \epsilon} \neq 0$, while the integral vanishes outside this interval. This completes the proof, assuming that the partial derivative (i.e., the above limit) exists.\footnote{It is possible for the $\gamma-$function to have sudden changes (e.g., when the surface has spherical patches centered at the query point), in which case it becomes non-differentiable with respect to $r = |\eta/\xi|$. However, the formulation remains valid if we use the Dirac delta function to handle differentiation at such discontinuities.}
\end{proof}
Substituting for the $\phi-$kernel from (\ref{eq_form_15}) in (\ref{eq_form_17A}) gives the following, noting that $\lambda(\zeta)$ is a function of $\xi(\mathbf{p}, S)$ only, hence fixed over the course of integration:
\begin{equation}
    \rho_\sigma(\mathbf{p}, S) = \frac{\lambda}{2\pi\sigma} \int_1^{|\frac{\mu}{\xi}|} \gamma_{r-1}' \frac{e^{-\frac{1}{2}(r - 1)^2/\sigma^2 - 2 \ii \tan^{-1}r}}{1 + r^{-2}}~ dr, \label{eq_form_17B}
\end{equation}
If $\gamma'_\epsilon$ is not bounded (e.g., at the center of perfectly spherical patches on the manifold) the integral is still convergent, but the aggregation abruptly changes as $(1+\epsilon)|\xi|$ passes through the radius of the spherical patch.

The truncated affinity function in (\ref{eq_form_17t}) can also be rewritten in terms of a 1D integral:
\begin{equation}
    \bar{\rho}_\sigma(\mathbf{p}, S) = \frac{\lambda}{2\pi\sigma} \int_1^{1 + \epsilon} \gamma_{r-1}' \frac{e^{-\frac{1}{2}(r - 1)^2/\sigma^2 - 2 \ii \tan^{-1}r}}{1 + r^{-2}}~ dr. \label{eq_form_17C}
\end{equation}
If $\sigma \ll 1$, the major contributions to the integral are over the subinterval $r \in [0, 1 + \epsilon]$, where $\epsilon = c\sigma$, which means over the contact region $C_\epsilon(\mathbf{p}, S)$ (see Figure \ref{figure9} (c) and (d)). The contributions from $r \in [1 + \epsilon, |\mu/\xi|]$ result in a negligible error, denoted earlier by $E_{\sigma, \epsilon}(\mathbf{p}, S)$. This is made precise as follows:
\begin{prop}
    If $\sigma \ll 1$, the truncation error $E(\mathbf{p}, S; \sigma) = \rho(\mathbf{p}, S; \sigma) - \bar{\rho}(\mathbf{p}, S; \epsilon)$ is bounded as follows:
    \begin{equation}
        |E| \leq E_m , \quad \text{if } \frac{\epsilon}{\sigma} \geq \left[ 2\log_e\left(\frac{|\lambda|}{2\pi \sigma} \frac{\gamma_m^+}{E_m} \right) \right]^{\frac{1}{2}}, \label{eq_form_18b}
    \end{equation}
    where $E = E_{\sigma, \epsilon}(\mathbf{p}, S)$, $E_m$ is the desired error upper-bound (arbitrarily chosen) and $\gamma_m^+ \geq \gamma_\infty(\mathbf{p}, S)$ is an upper-bound to the aggregated {\it unsigned} spatial angle of observation:
    \begin{equation}
        \gamma_m^+ \geq \int_1^{|\frac{\mu}{\xi}|} |\gamma_{r-1}'|~dr = \int_1^{|\frac{\mu}{\xi}|} \big|\frac{\partial~}{\partial \epsilon}\gamma_\epsilon \big|_{\epsilon = (r-1)} ~dr.
    \end{equation}
    The tightest error bound is obtained if $\gamma_m^+$ is picked equal to the right-hand side integral above (if known), while the most conservative choice is $\gamma_m^+ = 4\pi$, in which case the right-hand side of (\ref{eq_form_18b}) simplifies.
    Lastly, $\lambda \in \{+\lambda_1, -\lambda_2\}$ depending on $\mathrm{pmc}(\mathbf{p}, S)$ as a result of definition in {\rm (\ref{eq_form_16c})}.
\end{prop}
\begin{proof}
    The error is due to integration over $r \in [1 + \epsilon, |\frac{\mu}{\xi}|]$:
    \begin{align*}
        E &= \frac{\lambda}{2\pi\sigma} \int_{1+\epsilon}^{|\frac{\mu}{\xi}|} \gamma_{r-1}' \frac{e^{-\frac{1}{2}(r - 1)^2/\sigma^2 - 2 \ii \tan^{-1}r}}{1 + r^{-2}}~ dr, \\
        \Rightarrow |E| &\leq \frac{|\lambda|}{2\pi\sigma} \int_{1+\epsilon}^{|\frac{\mu}{\xi}|} \frac{ |\gamma_{r-1}'| }{1 + r^{-2}} ~e^{-\frac{1}{2}(r-1)^2/\sigma^2}~dr, \\
        \Rightarrow |E| &\leq \frac{|\lambda|}{2\pi\sigma} e^{-\frac{1}{2} (\frac{\epsilon}{\sigma})^2} \int_{1+\epsilon}^{|\frac{\mu}{\xi}|} |\gamma_{r-1}'| ~dr = \frac{|\lambda|}{2\pi\sigma} e^{-\frac{c^2}{2}} \gamma_m'.
    \end{align*}
    Noting that $r \geq (1 + \epsilon) \geq 1$ and $(1 + r^{-2})^{-1} \leq 1$ and $e^{-\frac{1}{2}(r-1)^2/\sigma^2} \leq e^{-\frac{1}{2}(\frac{\epsilon}{\sigma})^2}$. Equating the right-hand side to $E_m$, and solving for $c = \epsilon/\sigma$ yields (\ref{eq_form_18b}).
\end{proof}

Using (\ref{eq_form_17C}) instead of (\ref{eq_form_17B}) at a small cost bounded by (\ref{eq_form_18b}) enables the algorithm to perform the integration only over the $\epsilon-$ANNs of the query point on the manifold and neglect the contribution of the distant elements, significantly speeding up the implementation explained in Section \ref{sec_implement} with little compromise in the accuracy.

The relationship between the SDF and the conventional notions of shape skeleton (such as medial axis or its variants) becomes more intuitive when $\sigma \rightarrow 0^+$, in which case the Gaussian bell approaches the Dirac delta function: $\lim_{\sigma \rightarrow 0^+} g_\sigma(x) = \delta(x)$. Consequently, the integrand in (\ref{eq_form_17}) vanishes except when $\zeta(\mathbf{p}, \mathbf{q}, S)$ lies along the bisector $\eta = |\xi|$ in the complex plane. The integral thus degenerates to one over the contact region $C_0(\mathbf{p}, S)$, which is generally composed of $0-$, $1-$ and $2-$ dimensional regions. This means that the support of the SDF shrinks to a subset of the medial axis, as $\sigma \rightarrow 0^+$.
A rigorous analysis of the degenerate SDF and its relationship with the medial axis, especially reasoning about the dimensionality of $C_0(\mathbf{p}, S)$, is beyond the scope of this article.

\begin{table}
    \caption{The dominant overlaps for shape complementarity} \label{table_1}
    \begin{tabular}{p{2.0cm} p{2.0cm} p{2.8cm}}
        \hline \\[-2ex]
        $\mathbf{p}$ {\rm versus} $S_1$ & $\mathbf{p}$ {\rm versus} $\tau S_2$ & {\rm Score contribution} \\
        \\[-2ex] \hline
        $\mathbf{p} \in \mathbf{e}(S_1)$ & $\mathbf{p} \in \mathbf{i}(\tau S_2)$ & $\propto (\mp \ii\lambda_1) (\pm \ii\lambda_2)$ \\
        (high-density) & (high-density) & $= +\lambda_1 \lambda_2 = +\mathcal{O}(1)$ \\
        \hline
        $\mathbf{p} \in \mathbf{i}(S_1)$ & $\mathbf{p} \in \mathbf{e}(\tau S_2)$ & $\propto (\pm \ii\lambda_2) (\mp \ii\lambda_1)$ \\
        (high-density) & (high-density) & $= +\lambda_2 \lambda_1 = +\mathcal{O}(1)$ \\
        \hline
        $\mathbf{p} \in \mathbf{i}(S_1)$ & $\mathbf{p} \in \mathbf{i}(\tau S_2)$ & $\propto (\pm \ii\lambda_2) (\pm \ii\lambda_2)$ \\
        (high-density) & (high-density) & $= -\lambda_2^2 = -\mathcal{O}(\mathfrak{p})$ \\
        \hline
        $\mathbf{p} \in \mathbf{e}(S_1)$ & $\mathbf{p} \in \mathbf{e}(\tau S_2)$ & $\propto (\mp \ii\lambda_1) (\mp \ii\lambda_1)$ \\
        (high-density) & (high-density) & $= -\lambda_1^2 = -\mathcal{O}(\mathfrak{p}^{-1})$ \\
        \hline
    \end{tabular}
\end{table}
\begin{figure}
    \centering
    \includegraphics[width=0.48\textwidth]{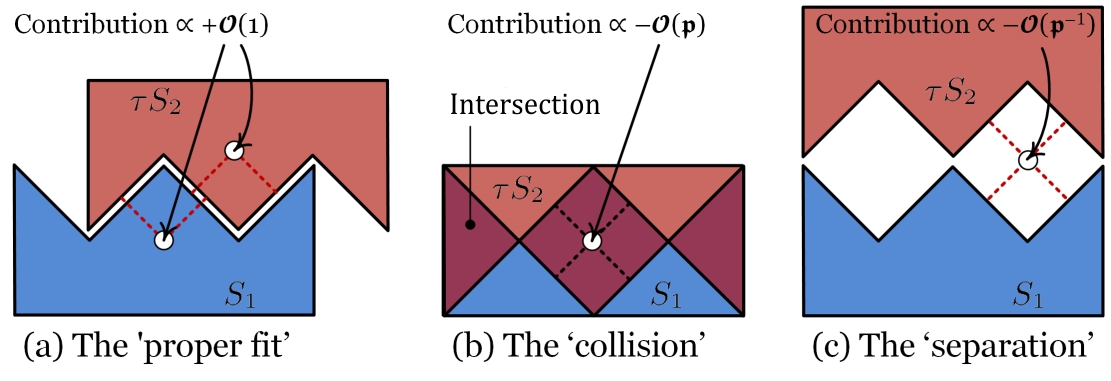}
    \caption{The spatial relations between the assembly partners and the contributions of high-density points, listed in Table \ref{table_1}.} \label{figure10}
\end{figure}

There is another reason behind the inverse-square form of the complex function $\phi_P(\zeta) = 1/(\sqrt{2\pi} \zeta^2)$, which is the resulted phase shift $\angle \phi = -2 \angle \zeta$. This in turn, together with an assignment of different coefficients $+\lambda_1 \neq -\lambda_2$ assigned to exterior and interior, respectively, results in phase and amplitude separations that are useful for a meaningful scoring scheme, explained here in simple terms using the schematics in Figure \ref{figure10}.
\begin{itemize}
    \item If $\mathbf{p}$ is an exterior point of $S$ (denoted $\mathbf{p} \in \mathbf{e}(S)$ thus $\xi > 0$), with high medial density, the complex spread $\zeta(\mathbf{p}, \partial S, S)$ is highly populated around $\angle \zeta = \pm \pi/4$.
    \item If $\mathbf{p}$ is an interior point of $S$ (denoted $\mathbf{p} \in \mathbf{i}(S)$ thus  $\xi < 0$), with high medial density, the complex spread $\zeta(\mathbf{p}, \partial S, S)$ is highly populated around $\angle \zeta = \pm 3\pi/4$.
\end{itemize}
Therefore, $\angle \phi \approx \mp \pi/2$ and $|\phi| = \lambda_1 \phi_M \phi_P$ for external points and $\angle \phi \approx \mp 3\pi/2 \equiv \pm \pi/2$, and $|\phi| = \lambda_2 \phi_M \phi_P$ for internal points near the medial axis---assuming the same intensities for the $\phi_M$ and $\phi_P$ terms.
The two cases differ by a phase difference of $\pi$ and amplitude ratio of $\propto \lambda_2/\lambda_1$. If $\sigma \ll 1$, the medial components become dominant, hence the affinity functions obtained for the two cases by accumulating the $\phi(\zeta)$ values in (\ref{eq_form_13}) are expected to have similar relationships, i.e., $\rho(\mathbf{p}, S) \propto \mp \ii \lambda_1$ for the high-density exterior points, while $\rho(\mathbf{p}, S) \propto \pm \ii \lambda_2$ for the low-density interior points.

Given two solids $S_1$ and $S_2$ and a configuration $\tau \in \mathrm{SE}(3)$, the three types of overlaps that dominate the shape complementarity score in (\ref{eq_form_8}) are depicted in Table \ref{table_1}, assuming $\lambda_1 = \mathcal{O}(\mathfrak{p}^{-1})$ and $\lambda_2 = \mathcal{O}(\mathfrak{p})$, where $\mathfrak{p}>1$ is referred to as the `penalty factor'. Notice that exterior-interior overlaps impart a {\it reward} of $+\mathcal{O}(1)$ (proper fit), while interior-interior overlaps add a large {\it penalty} of $-\mathcal{O}(\mathfrak{p})$ (collision), and the exterior-exterior overlaps add a small penalty of $-\mathcal{O}(\mathfrak{p}^{-1})$ (separation), illustrated in Figure \ref{figure10}.


\section{Implementation} \label{sec_implement}

There are various possible geometric representations including but not limited B-reps. Our method is independent of the choice of representation scheme, as long as it supports integration of the $\phi-$kernel over the boundary. Here we present the implementation for triangulated B-reps given in the form of triangular tessellations.

Given a mesh approximation $\Delta_n(S)$ of the $2-$manifold boundary $\partial S$ of a solid $S \in \mathcal{P}_\ast(\mathds{R}^3)$ composed of $n$ triangles, the $\zeta-$map and affinity integral can be computed for a given query point $\mathbf{p} \in \mathds{R}^3$ in $O(n)$ basic steps, having precomputed and tabulated the integral in (\ref{eq_form_17B}) for parameterized geometry of an arbitrary triangle and its vertex distances against a point, normalized by the minimum distance $|\xi(\mathbf{p}, S)|$. The evaluation can be carried out on a uniform 3D grid $G_m$ of $m$ sample points, which requires a total time of $O(mn)$. However, we showed in Section \ref{sec_form} that the affinity function can be approximated by the truncated integral over the subset of triangles $\Delta_s(\mathbf{p}) \subset \Delta_n$ that carry at least one vertex that is an $\epsilon-$approximate nearest neighbor to the query point $\mathbf{p} \in G_m$. Hence there is a possibility of speeding the algorithm up by leveraging efficient all-approximate nearest neighbor search algorithms.

Given the B-reps of two assembly partners $S_1$ and $S_2$, each as a mesh $\Delta_{n_1}(S_1)$ and $\Delta_{n_2}(S_2)$, and their affinities computed over the uniform grids $G_{m_1}$ and $G_{m_2}$, a cascade evaluation of the convolution in (\ref{eq_form_8}) for a finite collection of transformations $T \subset \mathrm{SE}(3)$, takes $O(m|T|)$ time, assuming $m_{1,2} = O(m)$. This can be improved by the decomposition $T = T_r \times T_t$ where $T_r \in \mathrm{SO}(3)$ is the rotation set and $T_t \in \mathds{R}^3$ is the translation set. For each rotation $\tau_r \in T_r$, an FFT can be used for cumulative evaluation of the convolution for all $\tau_t \in T_t$ in $O(m \log |T_t|)$ assuming that $|T_t| \ll m$. The total time for evaluating all $|T| = |T_r| |T_t|$ poses then becomes $O(m |T_r| \log |T_t|)$, which is a significant improvement over the cascade evaluation time of $O(m |T_r| |T_t|)$.

It is noticed that the SDF-based affinity function has higher variations around the points that are closer to the shape skeleton, hence the equispaced sampling on the uniform grid $G_m$ can be replaced with a nonequispaced sampling on a more efficient tree-based data structure, such as an octree $Q_{m'}$. The affinity data is adaptively sampled on the oct-tree, with deeper subdivisions at the high skeletal density regions, resulting in a significant reduction of space- and time-complexities due to $m' \ll m$. However, the traditional FFT algorithm \cite{Cooley1965} needs to be replaced with a grid-free nonequispaced FFT algorithm, such as the one introduced in \cite{Potts2001}, resulting in an arithmetic complexity of $O(\alpha m' |T_r| \log (\alpha |T_t|))$ for the oversampling factor $\alpha > 1$ (see \cite{Potts2001} for implementation details), which outperforms the grid-based FFT as long as $(m'/m) < \alpha$.


\section{Validation} \label{sec_validation}

\begin{figure}
    \centering
    \includegraphics[width=0.48\textwidth]{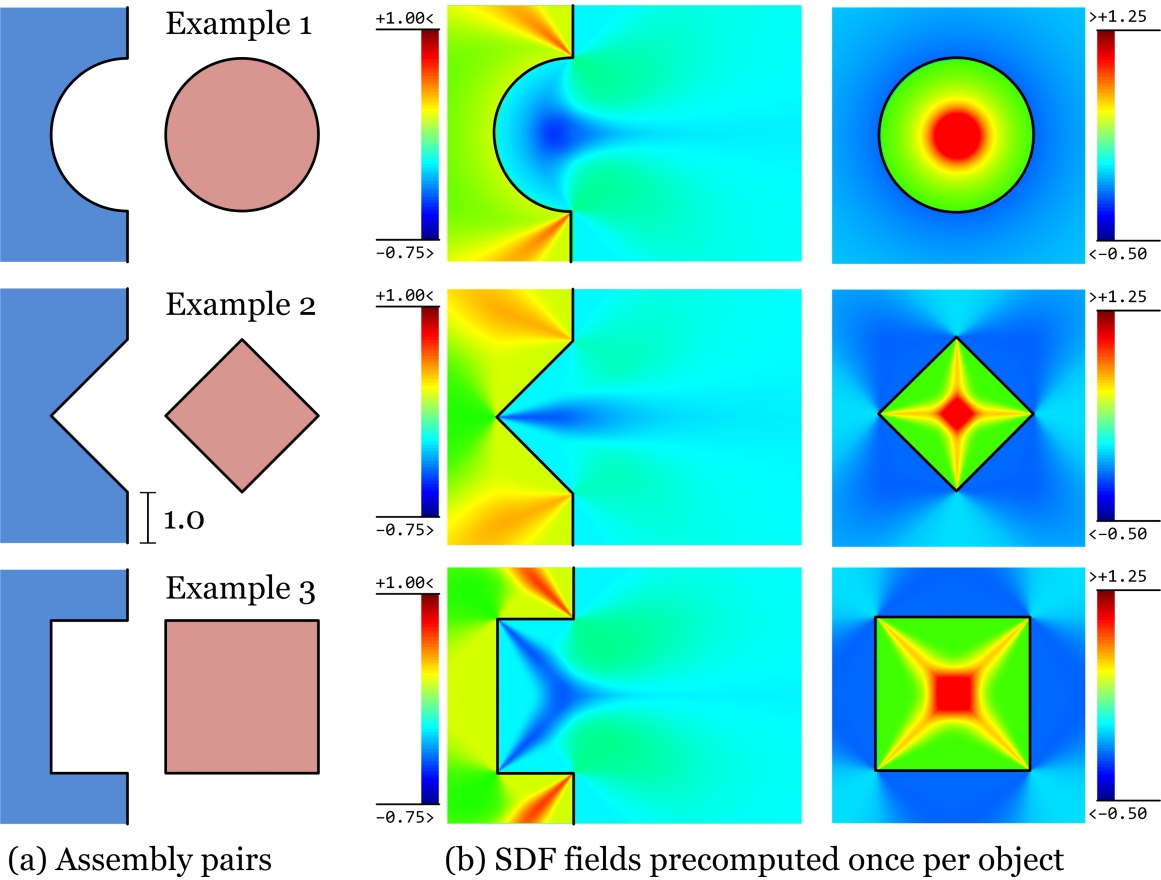}
    \caption{3 pairs of simple 2D solids are considered; the SDFs are precomputed offline, and used repeatedly to evaluate the shape complementarity score for different motions.} \label{figure11}
\end{figure}
\begin{figure}
    \centering
    \includegraphics[width=0.48\textwidth]{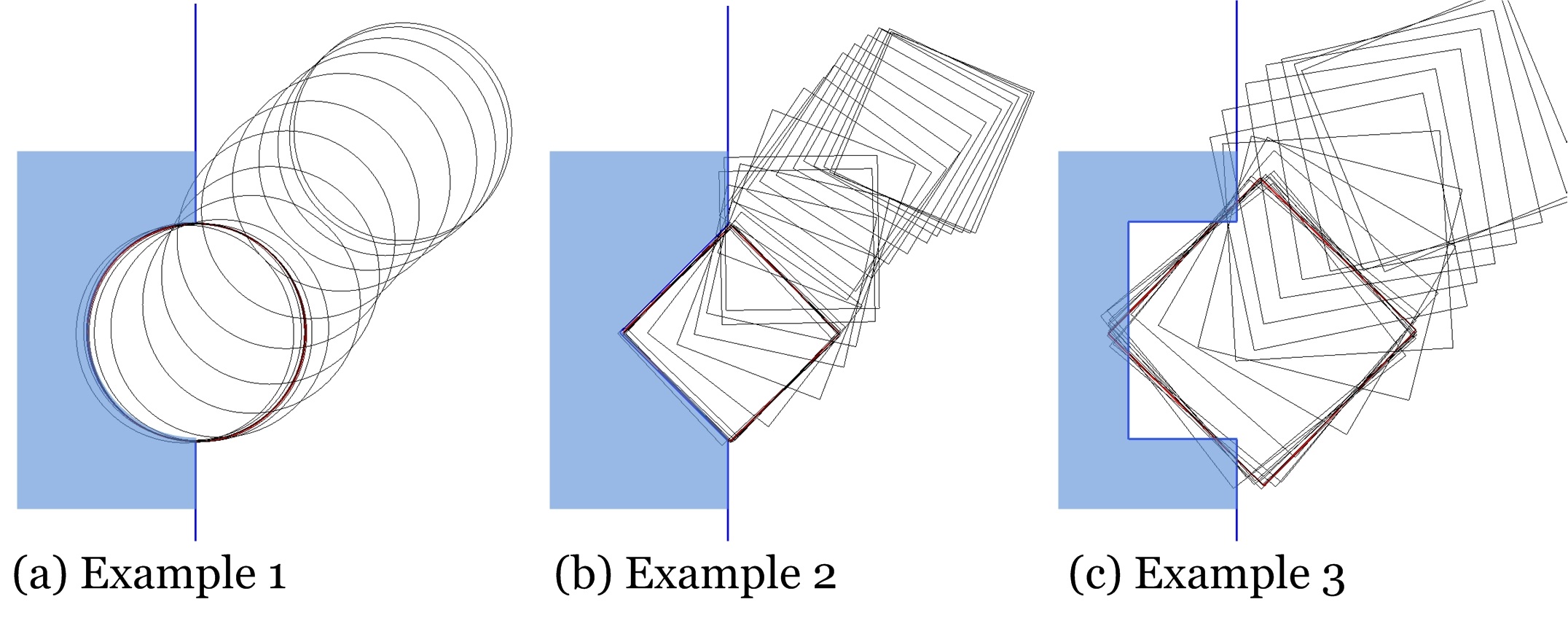}
    \caption{Snapshots of a dynamic simulation using the virtual geometric force-model $\mathbf{F}(\tau) = -\nabla E(\tau) \propto f(t, S_1, S_2)$. The moving object may get trapped in a local energy well, as in (c).} \label{figure12}
\end{figure}

The search for a near-optimal assembly configuration can be addressed as an optimization of the score function. To investigate the general properties of the affinity function in practice, we first consider 3 simple 2D example, shown in Figure \ref{figure11} (a). The circular arcs are approximated with line segments of at most $0.1$ length units, and the affinity integral in (\ref{eq_form_17C}) is numerically computed with an angular precision of $\delta \gamma = \eta^{-1} \cos \theta ~\delta r \leq 10^{-4}$ using Riemann sums. The affinity is computed over a uniform grid of query points with a resolution of $0.05$ units of length, using the coefficients $\lambda_1 = 1$, $\lambda_2 = 3$, a thickness factor of $\sigma = 0.5$ and a generous truncation factor of $\epsilon = 3\sigma = 1.5$. The resulting SDF affinity fields $\rho_{1,2}(\mathbf{p})$, plotted in Figure \ref{figure11} (b), are convolved for different motions $\tau \in \mathrm{SE}(2)$ represented by the 3-tuple $(x, y, \vartheta)$, where $(x,y)$ is the coordinates of the moving solid's centroid, and $\vartheta$ is its right-handed rotation around the positive $z-$axis.

\begin{figure}
    \centering
    \includegraphics[width=0.48\textwidth]{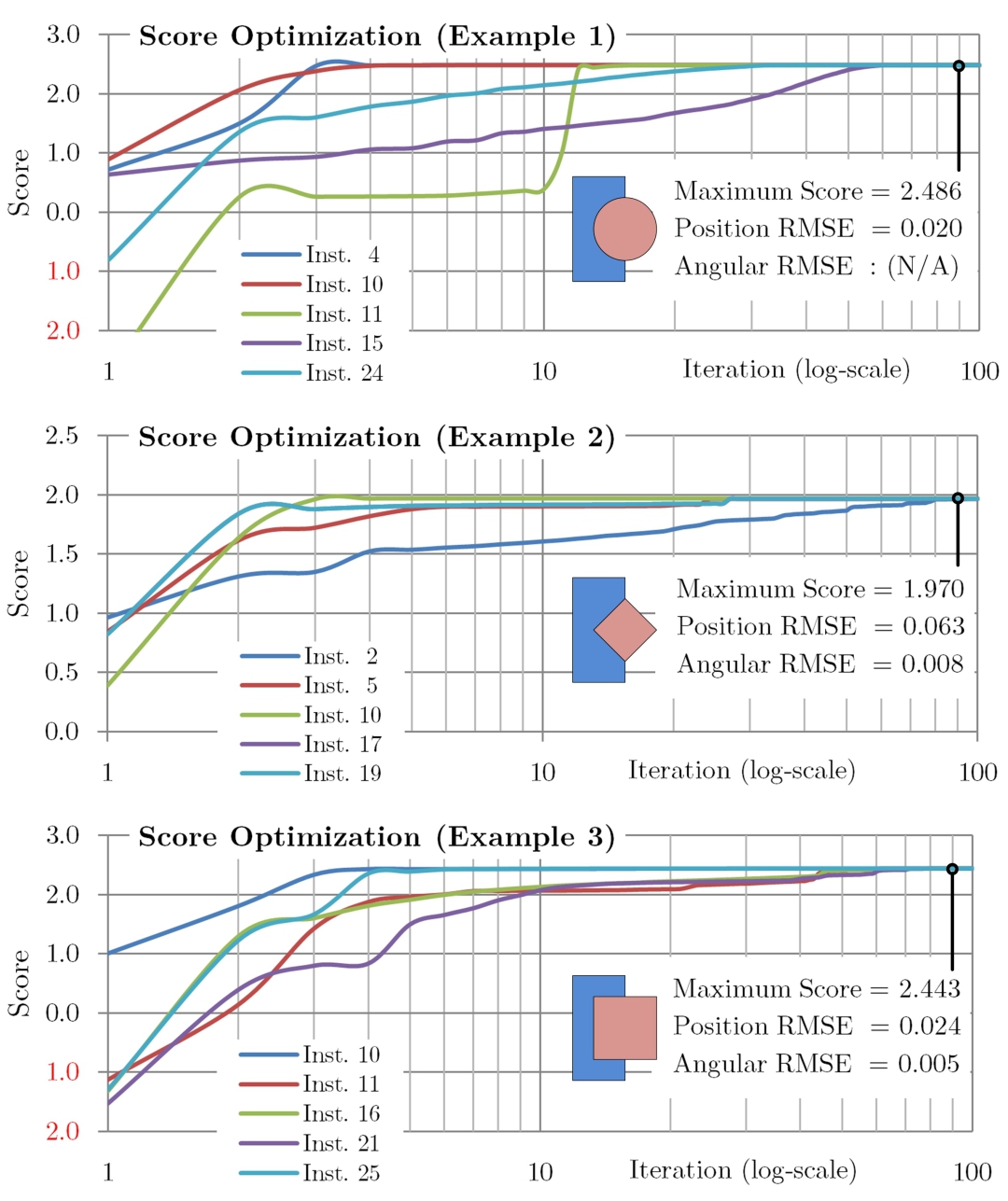}
    \caption{Multiple CG optimization progress of correct pairings, showing 5 instances out of 25 random starting configurations with highest local maxima.} \label{figure13}
\end{figure}

To validate the effectiveness of the continuous score function as a virtual energy-analogue in real-time applications, we first perform a rigid body dynamics simulation, assuming a potential force-field $\mathbf{F}(\tau) = -\nabla E(\tau)$ where $E(\tau) \propto -f(\tau, S_1, S_2)$ imposed by the stationary object $S_1$ on the moving object $S_2$, presented in Figure \ref{figure12} for a particular setting of inertia and damping coefficients. Such a virtual energy formulation conceives a purely geometric force-model, particularly useful for real-time applications such as semi-automatic solid assembly or user-interactive protein docking (e.g., using haptic devices). As expected from the formulation, the proximal component of affinity drags the moving object towards the binding site, where the medial component adjusts the orientation to maximize proper skeletal overlap. However, Figure \ref{figure12} (c) depicts the possibility of getting trapped in local minima where there is no user interference, as is inevitable with any gradient-based optimization as well. To avoid this, a collection of $|T_0| = 25$ random initial configurations $\tau_0 \in T_0$ are uniformly sampled in the range $(x, y, \vartheta)_0 \in [-2.5, +2.5]^2 \times [-\pi/4, +\pi/4]$. For different possible pairing of the parts in Figure \ref{figure11} (a), a conjugate-gradients (CG) maximization of the score $f(\tau)$ is carried out for $100$ iterations, starting from different initial configuration of the moving part,\footnote{This suffices for the simple 2D examples under consideration. However, a credible Monte Carlo search typically requires a larger sample size, and possibly larger number of iterations, in the case of more complicated shape features. For those cases, genetic algorithms can be used, assisted with FFT-based cumulative evaluation for large populations of samples in the configuration space, which will be presented elsewhere.} using a 2-point central-difference method for computing the gradients, with $\delta x = \delta y = 0.01$ units of length, and $\delta \vartheta = 0.01$ radians. For each example, the top $20\%$ (i.e., $5$ solutions out of $25$ which converged to the largest local maxima) are picked for inspection, the progress of which are plotted in Figure \ref{figure13} for correct partners, and in Figure \ref{figure14} for cross-pairings. For the proper pairings, the solution is within a translational RMSE of $\leq 0.024-0.063$ units of length (comparable to the grid resolution, hence reasonable) and a rotational RMSE of $\leq 0.002-0.008$ radians ($0.30^\circ-0.47^\circ$), of the correct assembly configuration known from visual perception. The RMSEs can exceed $0.15-0.32$ unit lengths and $0.17$ radians ($\sim 10^\circ$) for incorrect pairings, and is clearly correlated to the drops in the optimal scores, which are lower by at least $20-25\%$ compared to that of the correct pairing of the same constituents. Therefore, the method is useful in feature recognition and classification by comparing pairwise correlations.

Next, the effects of the SDF parameters on the score function and the RMSE are investigated with the help of Example 3 at the near-optimal configuration $\tau^\ast \approx \tau_0$, where $\tau_0$ is picked at the true global maximum of the score function (known by visual perception), and $\tau^\ast$ is obtained with 10 CG steps meant for relaxation, each time with a different set of SDF parameters. Figure \ref{figure15} (a) depicts the impact of the penalty factor $\mathfrak{p} = \lambda_2 / \lambda_1$, keeping $\lambda_1 \lambda_2$ constant. As expected from the definition, $f(\tau^\ast)$ does not change significantly with $\mathfrak{p} \geq 3.0$ at the collision-free configuration $\tau^\ast$ (see Table \ref{table_1} for an explanation). In addition, the translational RMSE decreases with increasing $\mathfrak{p}$, and remains in the range of $0.005-0.015$ units of length, which ensures the objectivity of the solution with respect to the choice of penalty ratio. However, the sensitivity of the score function to changes in $\tau$ increases with $\mathfrak{p}$, which can be characterized by the increase in $|\nabla^2 f(\tau^\ast)|$ (not shown here), due to higher intolerance to collision. $\mathfrak{p} \gg 1$ strictly prohibits collision, while slight penetrations might be allowed with $\mathfrak{p} \approx 1$ in the case of soft assembly (explained in Section \ref{sec_method}), hence the penalty factor is useful in tuning the flexibility of the assembly correlations, without significantly disturbing the optimal configuration.

\begin{figure}
    \centering
    \includegraphics[width=0.48\textwidth]{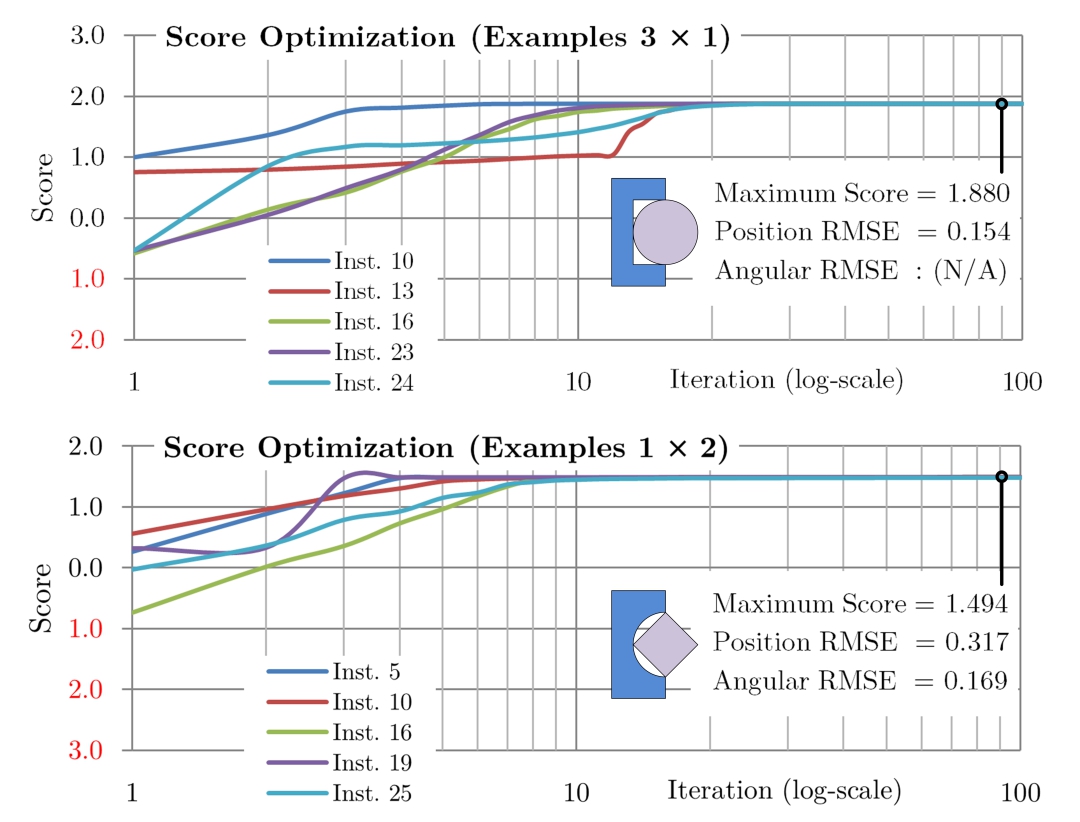}
    \caption{Multiple CG optimization progress of cross-pairings, showing 5 instances out of 25 random starting configurations with highest local maxima.} \label{figure14}
\end{figure}
\begin{figure}
    \centering
    \includegraphics[width=0.48\textwidth]{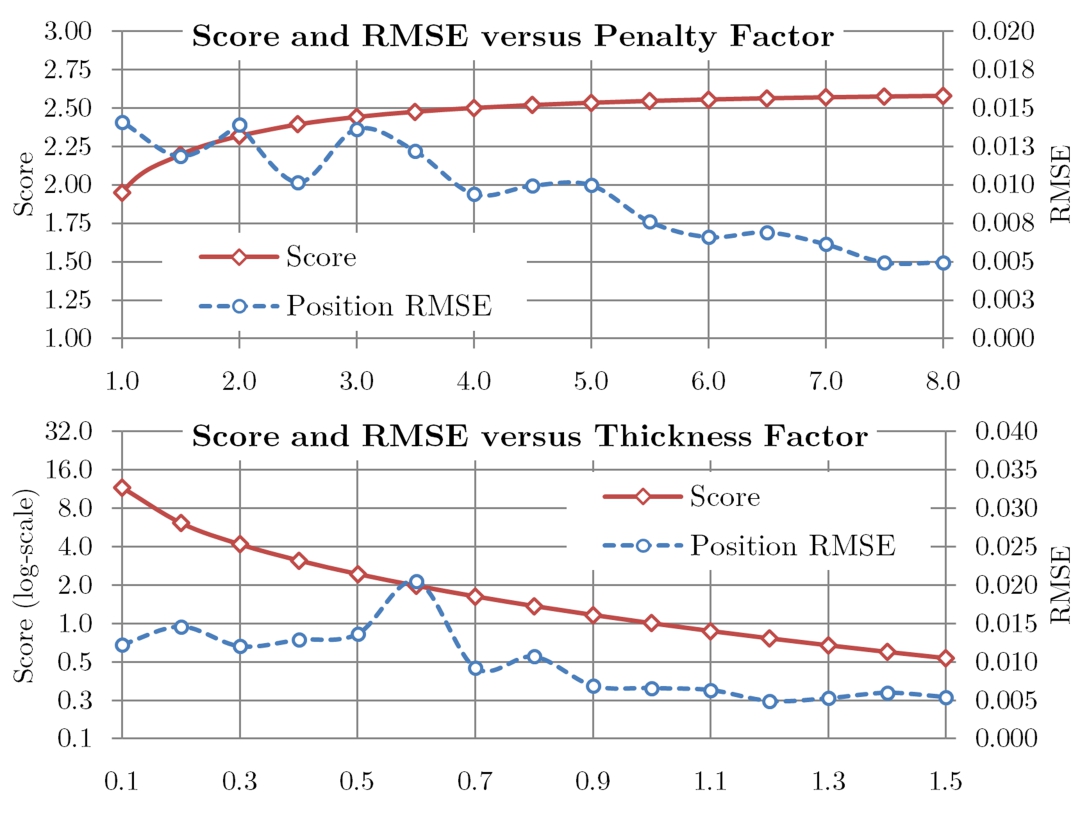}
    \caption{The effects of the SDF parameters $p = \lambda_2/ \lambda_1$ and $\sigma$ on the score and the RMSE.} \label{figure15}
\end{figure}

The central parameter to the medial effects in SDF is the thickness factor $\sigma$, the variation of which is investigated in Figure \ref{figure15} (b). As a result of definition, increasing $\sigma$ makes the Gaussian kernel sharper and more locally supported, decreasing the thickness of the medial regions. The result is higher sensitivity to deviations from optimal skeletal overlap, validated by the observed sharp increase in $f(\tau^\ast)$ with decreasing $\sigma$. Letting $\sigma \ll 1$ results in narrow and high-density medial regions, and allows precision assembly at the expense of a higher resolution requirement for the convolution, while $0.5 \leq \sigma$ allows low-resolution docking for imperfect or flexible contact. We argue that this provides a means to adjust the geometric {\it specificity} of assembly features, which can be defined as the differences in geometric {\it affinities} between possible pairings, i.e., the change in the complementarity score when one partner is replaced with another.\footnote{The terminology is borrowed from protein-protein interaction energetics, where specificity is critical for proper molecular recognition and binding.} As was the case with the penalty factor, the choice of the thickness factor has little effect on the position of the near-optimal solution $\tau^\ast$, noting that the RMSE remains in the range of $0.005-0.020$ over one order of magnitude change in the score $f(\tau^\ast) \sim 1/\sigma$. Therefore, the SDF parameters can be regulated to address the particular application needs (precision versus approximate assembly, rigid versus soft protein docking, broadening/narrowing the optimization pathway, etc.) without destabilizing the near-optimal solution.

\begin{figure}
    \centering
    \includegraphics[width=0.48\textwidth]{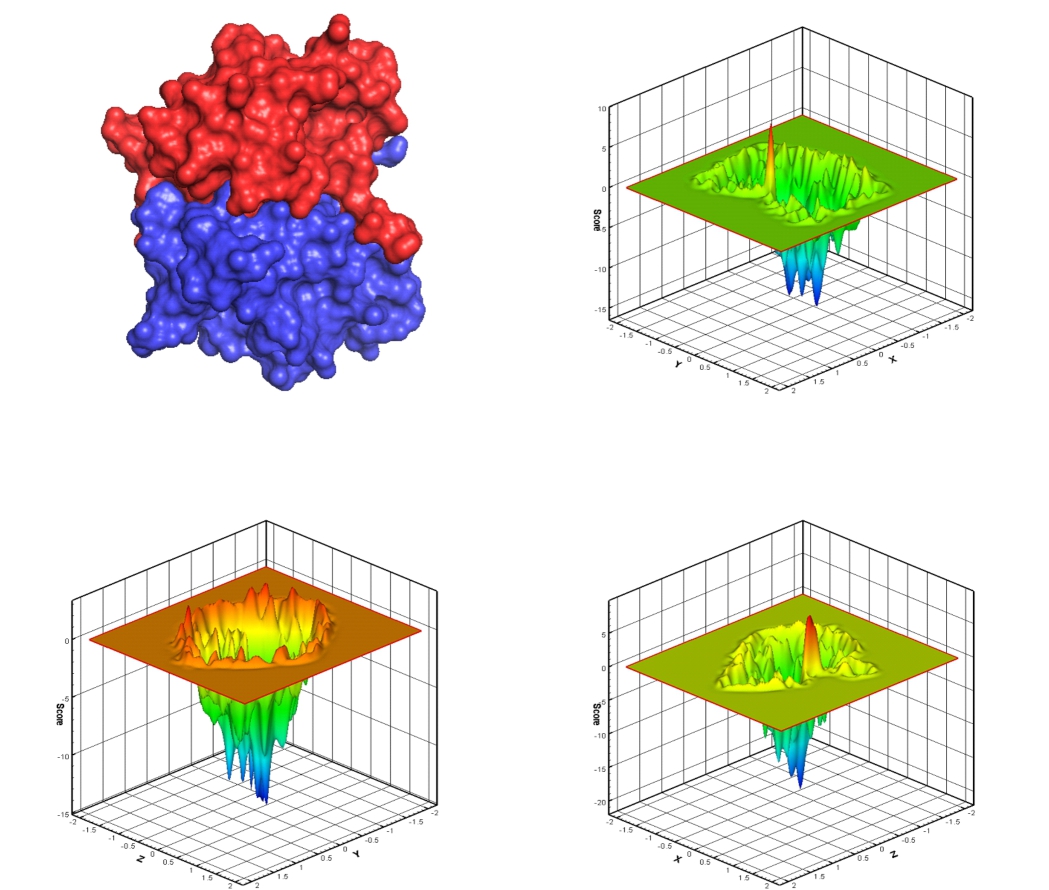}
    \caption{DSL-based shape complementarity score function variations for a pair of proteins (namely, NTF2 and GTPase Ran) (PDB Code: 1A2K) \cite{Stewart1998}.} \label{figure16}
\end{figure}
\begin{figure}
    \centering
    \includegraphics[width=0.48\textwidth]{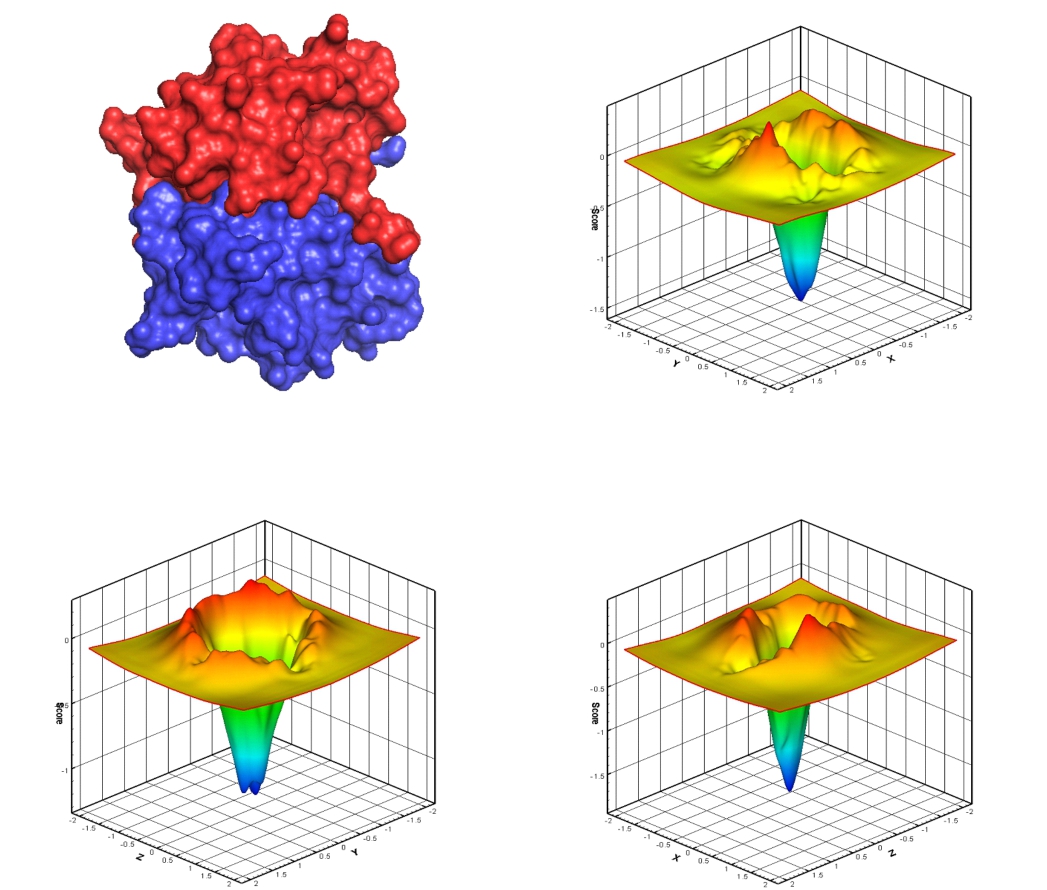}
    \caption{SDF-based shape complementarity score function variations for a pair of proteins (namely, NTF2 and GTPase Ran) (PDB Code: 1A2K) \cite{Stewart1998}.} \label{figure17}
\end{figure}

The next step is to validate the results for real 3D examples in solid assembly and protein-docking, the likes of which were shown in Figures \ref{figure1} and \ref{figure2} of the introduction.

Figures \ref{figure16} and \ref{figure17} plot the variations of the shape complementarity score function using the cross-correlation formula in (\ref{eq_form_8}) for a pair of protein domains taken from the X-ray crystallography of the protein complex formed by the nuclear transport factor 2 (NTF2) and the GTPase Ran (PDB Code: 1A2K) \cite{Stewart1998}. The plots show the shape complementarity score function changes for translations along $xy-$, $yz-$ and $zx-$ planes, for a fixed relative orientation (same as the native configuration). Figure \ref{figure16} uses the DSL-based formulation of the density function in Section \ref{sec_method} while Figure \ref{figure17} uses our SDF-based formulation. The cross-correlation of both density functions accurately predicts the best fit at the native configuration. However, the former gives a score function of relatively narrow support with undesirable sharp spikes that make the search for the global optimum challenging. On the other hand, our formulation has the advantage of having a nonzero score function over a larger subspace of the configuration space, as well as much smoother variations that are required for gradient-based optimization (e.g., CG-like methods).

Figure \ref{figure18} similarly illustrates the variations of our SDF-based score function for a mechanical assembly with nontrivial matching features as well as sharp and dull corners. Once again, the best fit configuration is predicted correctly irrespective of geometric complexity of the interacting surface patches, whose detection and matching using feature-based heuristics is otherwise challenging.

\begin{figure}
    \centering
    \includegraphics[width=0.48\textwidth]{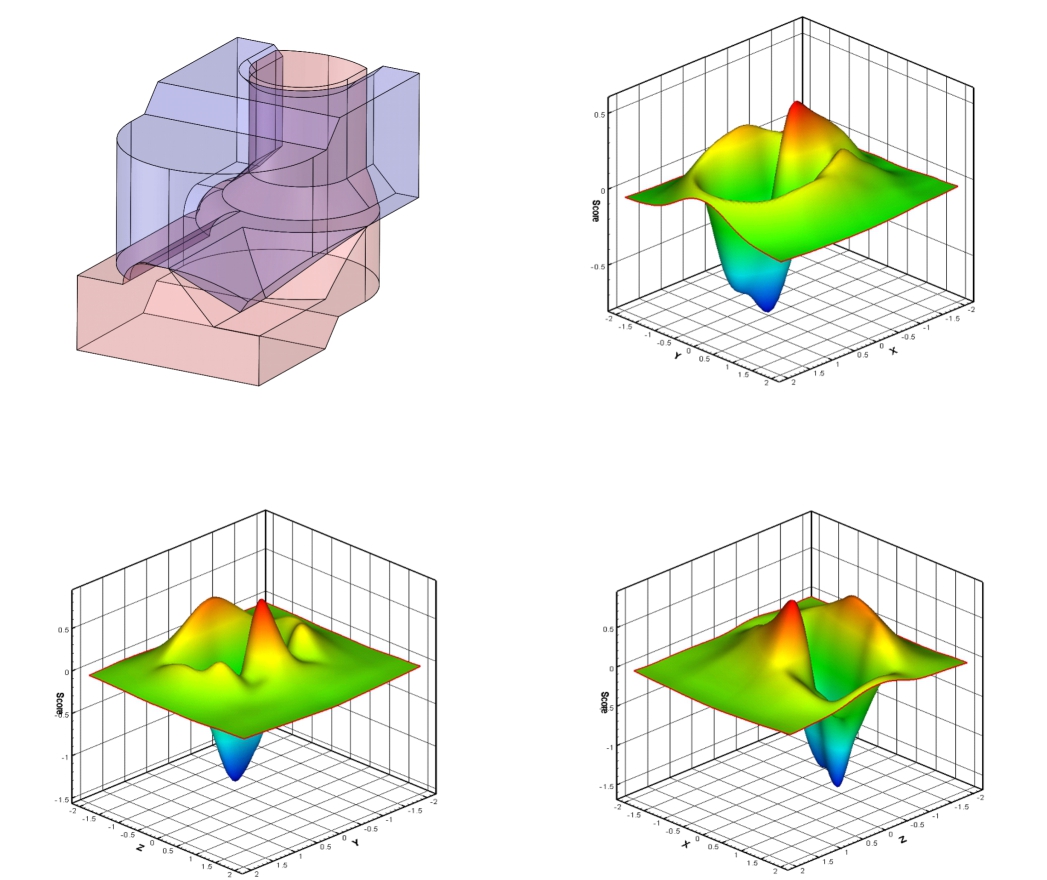}
    \caption{SDF-based shape complementarity score function variations for a pair of mechanical parts bounded by $C^1-$discontinuous algebraic surfaces.} \label{figure18}
\end{figure}
%


\section{Conclusion} \label{sec_conclusion}

A broad range of problems in geometric computing are dependent on a quantification of shape complementarity, which despite being intuitive to human perception for the most simple objects, is difficult to identify for general shapes. As we discussed in Section \ref{sec_method} in the context of the widely studied protein-docking problem, the existing techniques are mostly equipped with application-oriented heuristics, over-simplifications of the shapes, or restrictive assumptions on the building blocks in the first place. In Section \ref{sec_form} we proposed a theoretical framework to address the shape complementarity problem for objects of arbitrary shape in 2D and 3D, which can be extended in principle to higher dimensional spaces.
We argued that the first step in the common paradigm of convolution-based score-and-search, which is the proper formulation of the affinity function, essentially requires a projection of the shapes into a lower-dimensional space (i.e., the complex plane). We decomposed this process into two independent steps, first transforming the shape representation into an equivalent form in terms of Euclidean distances from the boundary elements to different observation points, projected to the complex plane ($\zeta-$mapping), followed by integrating a kernel $\phi: \mathds{C} \rightarrow \mathds{C}$ that extracts the relevant features of distance distribution to complementarity analysis. We argued that ostensibly different approaches to formulate the affinity function can be unified under the same conceptualization, differences being attributed to their choices of the $\phi-$kernel.

Furthermore, we proposed a particular kernel formulation that subsequently led to an implicit redefinition of the shape skeleton as a space-continuous density function, to which we referred as the skeletal density function (SDF). The $\phi-$kernel was devised to impart two effects into the SDF; namely, a proximal effect due to the inverse-square component, the choice of which appeared natural for a density function in 3D supported by the discussion on spatial angles; and a medial component that assigns higher densities to the query points with more extensive $\epsilon-$approximate nearest neighbor region on the boundary. We demonstrated in Section \ref{sec_validation} that the proposed score model is consistent with our perception of a proper fit; it favors alignment of complementary sharp features, while avoids collision or separation, hence appears a strong candidate to replace the heuristics for algorithm design. We also showed that the branched topologies of the support region of the overlapped SDFs provide a pathway for the optimization algorithm (i.e., a low-energy valley, making the analogy between score maximization and energy minimization) with relevant local clues towards the near-optimal solution. This in turn enables the simplest gradient-based algorithms to converge in a few number of iterations, as validated in Section \ref{sec_validation}.

In order to make the framework available for use in user-interactive real-time applications, more efficient data structures and algorithms need to be exploited, including all-approximate nearest neighbor search algorithms for faster computation of affinities, systematic adaptive sampling techniques, and nonequispaced FFTs for cumulative evaluation of optimization trials.


\bibliographystyle{plain}
\protect\footnotesize \bibliography{CDL-TR-14-01}

\end{document}